\documentclass{llncs}

\usepackage{amsmath,amssymb,color,framed,float,microtype,mathrsfs,enumerate,tikz}
\usetikzlibrary{shapes, arrows, calc, arrows.meta, patterns} %TODO put back in
\usepackage[vlined, boxed, linesnumbered]{algorithm2e}
\usepackage[hidelinks]{hyperref}
\usepackage[margin=1.00in]{geometry}
\DontPrintSemicolon
\SetStartEndCondition{ }{}{}
\SetKwFor{For}{for}{\string:}{}
\SetKwFor{Loop}{repeat}{\string:}{}
\SetKwIF{If}{ElseIf}{Else}{if}{:}{else if}{else:}{}
\SetKwFor{While}{while}{:}{}
\SetKwInOut{Input}{input}
\SetKwInOut{Output}{output}
\SetArgSty{}

\newenvironment{prop}{\begin{proposition}}{\end{proposition}}
\newenvironment{thm}{\begin{theorem}}{\end{theorem}}
\newenvironment{lem}{\begin{lemma}}{\end{lemma}}

\spnewtheorem{defn}[definition]{Definition}{\bfseries}{}

\newcommand{\R}{\mathbb{R}}

\newcommand{\N}{\mathbb{N}}

\renewcommand{\O}{\mathcal{O}}
\newcommand{\bkp}{\texttt{bkp}}

\renewcommand{\epsilon}{\varepsilon}
\renewcommand{\tilde}{\widetilde}
\renewcommand{\bar}{\overline}

\DeclareMathOperator{\capacity}{Cap}
\DeclareMathOperator{\demand}{Dem}

\DeclareMathOperator{\rec}{REC}
\DeclareMathOperator{\rs}{RS}

\newcommand{\suppress}[1]{}

\tikzset{%
  truecircle/.style = {circle, minimum size = 1cm, draw, thick},
  trueellipse/.style = {ellipse, minimum height=3cm, minimum width=1cm, draw, thick},
  falsecircle/.style = {circle, minimum size = 1cm, draw, color=red, dashed, thick},
  falseellipse/.style = {ellipse, minimum height=3cm, minimum width=1cm, draw,color=red, dashed, thick},
  confusedellipse/.style = {ellipse, minimum height=3cm, minimum width=1cm, draw, double, thick},
  trueedge/.style = {-{Latex}, thick},
  falseedge/.style = {-{Latex}, color=red, dashed, thick},
  truehalfedge/.style = {-{Latex[left, line width=0]}, thick},
  falsehalfedge/.style = {-{Latex[right, line width=0]}, color=red, dashed, thick},
  properredge/.style = {-{Latex}, thick},
  truesplitedge/.style = {-, thick},
  falsesplitedge/.style = {-, color=red, dashed, thick}
}

\renewenvironment{proof}{\begin{origproof}}{\qed \end{origproof}}

\title{
The Adversarial Noise Threshold for Distributed Protocols
}
\author{William M. Hoza\thanks{Supported by a Nellie Bergen and Adrian Foster Tillotson Summer Undergraduate Research Fellowship from the California Institute of Technology, as well as by the ARCS Los Angeles Founder Chapter.}%
\and Leonard J. Schulman\thanks{Supported in part by NSF Award 1319745.}}
\institute{Caltech, Pasadena, CA 91125 \\ \email{\{whoza, schulman\}@caltech.edu}}
\date{}

\begin{document}
\maketitle

\begin{abstract} \setlength\parindent{16pt}
We consider the problem of implementing distributed protocols, despite adversarial channel errors, on synchronous-messaging networks with arbitrary topology.

In our first result we show that any $n$-party $T$-round protocol on an undirected communication network $G$ can be compiled into a robust simulation protocol on a sparse ($\O(n)$ edges) subnetwork so that the simulation tolerates an adversarial error rate of $\Omega\left(\frac{1}{n}\right)$; the simulation has a round complexity of %$\O\left(\frac{m \log n}{n} T\right)$ \wmh{Technically, $T$ outside the $\O(\cdot)$ would make the statement slightly stronger, since it would assert that the round complexity is linear in $T$ on all graphs, instead of just on suff. large graphs. Right?}\ljs{Maybe this is  a distinction only for $T=0$?} \wmh{No. If the round complexity were, say, $\frac{m \log n}{n} T^2$ whenever $n < 100$ and $\frac{m \log n}{n} T$ whenever $n \geq 100$, then technically, the round complexity would be $\O(\frac{m \log n}{n} T)$, if we follow e.g. the definition on Wikipedia (\url{http://en.wikipedia.org/wiki/Big_O_notation\#Multiple_variables}). But that round complexity would not be $\O(\frac{m \log n}{n})T$.}\ljs{Wikipedia is usually accurate but perhaps in this case it does not represent the consensus notation. I believe most take it to mean "exists c s.t. forall but finitely many", the latter existing just to exclude zeros. My textbooks are curiously silent except for Matousek, Lectures on Discrete Geometry, which agrees with my usage. The advantage is that an O() runtime claim is now more meaningful: you can believe it for all values of the parameters.}
% WMH: resolved
$\O\left(\frac{m \log n}{n} T\right)$, where $m$ is the number of edges in $G$. (So the simulation is work-preserving up to a $\log$ factor.) The adversary's error rate is within a constant factor of optimal. Given the error rate, the round complexity blowup is within a factor of $\O(k \log n)$ of optimal, where $k$ is the edge connectivity of $G$.
We also determine that the maximum tolerable error rate on \emph{directed} communication networks is $\Theta(1/s)$ where $s$ is the number of edges in a minimum equivalent digraph.

Next we investigate adversarial \emph{per-edge error rates}, where the adversary is given an error budget on each edge of the network. We determine the limit for tolerable per-edge error rates on an arbitrary directed graph to within a factor of $2$. However, the construction that approaches this limit has exponential round complexity, so we give another compiler, which transforms $T$-round protocols into $\O(mT)$-round simulations, and prove that for polynomial-query black box compilers, the per-edge error rate tolerated by this last compiler is within a constant factor of optimal.
\end{abstract}
\thispagestyle{empty}
%\newpage
%\setcounter{page}{1}
\section{Introduction}
We consider the problem of protecting distributed protocols from channel noise. The two-party case has received extensive attention, which we briefly survey in Section~\ref{sec:prior}. The multiparty case has been studied in three works. Rajagopalan and Schulman~\cite{rs94} showed how to protect synchronous distributed protocols on digraphs with $m$ edges and $n$ vertices against stochastic noise (at a constant noise rate per bit transmission), slowing down by a factor of $\log(\text{max degree})$. 
Gelles, Moitra, and Sahai improved on these results \cite{gms11} by constructing a computationally efficient simulation with the same properties. The first study of adversarial noise on multiparty ($n>2$) networks is by 
Jain, Kalai, and Lewko \cite{JKL15}. They focused on a ``sequential'' communication model, in which there is at most one message in-flight in the network at any time. Their networks are undirected (which throughout this work we equate with a symmetric or bidirected digraph), and they show that if the graph 
contains one party who is connected to every other (a star subnetwork), then every ``semi-adaptive'' $T$-round protocol can be compiled into an $\O(T)$-round simulation protocol which tolerates an adversarial bit error rate of $\Omega(\frac{1}{n})$. They point out that this error rate is within a constant factor of optimal, because with an error budget of this order, the adversary can effectively cut off one party from the rest of the graph. They also prove another negative result, showing that in a certain black-box model, even if the adversary is restricted to a separate budget of errors for each party's outgoing messages, no constant error rate can be tolerated.

We return in this paper to the model of synchronous distributed protocols---in each unit of time, each party transmits one bit to each of its out-neighbors, as in~\cite{rs94}---but, unlike~\cite{rs94} and~\cite{gms11}, we treat adversarial error. Specifically, the adversary is assumed to know the inputs to all the parties, and the entire history of communications up to the present. Only the private randomness of the parties is unknown to the adversary. Our primary objective is to determine (up to a constant) the noise threshold at which reliable communication becomes possible. On undirected networks, we provide simulation protocols which achieve this threshold and which are within a factor of $\O(k \log n)$ of optimal in round complexity (for protocols achieving the threshold), where $k$ is the edge connectivity of the network.

\subsection{Outline of our results}

The starting point for our main result is a slight variant of the compiler constructed by Rajagopalan and Schulman \cite{rs94}, which we refer to as the \emph{RS compiler} (see Appendix~\ref{apx:rs} for the modifications). Previously this compiler was analyzed for stochastic errors. We show (Proposition~\ref{prop:rs94}) that the RS compiler tolerates an adversarial error rate of $\Omega(\frac{1}{m})$. On networks with bounded edge connectivity, if we only consider simulation protocols which run on the same networks as the original protocols, this is within a constant factor of the best possible error rate: with an error budget of this order, the adversary can effectively disconnect the network. Thus, to tolerate a higher error rate, we are forced to consider simulations running on subnetworks. 
Note that it would not suffice for the parties to simply send dummy messages on those edges which they do not want to use; rather, our model explicitly allows simulations to run on subnetworks. 
It may seem strange that turning off edges can help with noise resiliency but the key is that we are able to redesign the protocol so that, informally, it ``relies on all remaining edges evenly''; consequently, the adversary's most effective attacks, which apply the entire error budget to a small region, have an advantage factor of only $n$ rather than $m$. In outline, we achieve this as follows, given an arbitrary protocol on an undirected communication network: 
\begin{enumerate}[(a)]
\item We use multicommodity flow methods to route the messages of the original protocol through a cut sparsifier.
\item We modify the sparse network by adding back in some of the edges which were removed, so that the routes can be short in addition to having low congestion.\item We apply the RS compiler to this new protocol on the second sparse subnetwork, so that the final simulation tolerates an error rate of $\Omega(\frac{1}{n})$.
\end{enumerate}
This error rate is within a constant factor of optimal, as noted above. Furthermore, the round complexity blowup is within a factor of $\O(k \log n)$ of optimal (for protocols tolerant to this error rate), where $k$ is the edge connectivity of the graph on which the original protocol ran (Theorem~\ref{thm:runtime-lower-bound}.) (If one permits shared randomness, there are cases in which this gap can be narrowed, as we describe in Theorem~\ref{thm:magi}.)

The same basic strategy allows us to determine the optimal error rate on \emph{directed} graphs. We say that two digraphs on the same vertex set are \emph{reachability-equivalent} if they have the same reachability relation. A \emph{minimum equivalent digraph} of $G$ is a reachability-equivalent subgraph with the fewest possible edges. (See \cite{mt69,hsu75}.) We show (Theorem~\ref{thm:reachability-preserving-subgraph-error-rate}) that any protocol on an arbitrary digraph can be simulated to tolerate an error rate of $\Omega(\frac{1}{s})$, where $s$ is the number of edges in each minimum equivalent digraph. We also show (Theorem~\ref{thm:error-rate-negative-result}) that this error rate is within a constant factor of optimal.

We also investigate a more restricted adversary, who has a separate budget of errors for each edge. We prove (Theorems~\ref{thm:signal-diameter-positive-result} and~\ref{thm:signal-diameter-negative-result}) that the cutoff for tolerable per-edge error rates is $\Theta(\frac{1}{D})$, where $D$ is the maximum finite directed distance between any two parties in the digraph. However, the positive side of that argument involves a simulation with exponential round complexity. We prove (Theorem~\ref{thm:range-positive-result}) that there is a compiler which tolerates a per-edge error rate of $\Omega(\frac{1}{R})$, where $R$ is the maximum number of distinct vertices visited in any walk through the graph; the simulations output by that compiler have round complexity $\O(mT)$. The proof of Theorem~\ref{thm:range-positive-result} mostly consists of extending the arguments in \cite{rs94} to establish a tighter analysis of the RS compiler. By a similar argument to that used in \cite{JKL15}, we prove that this per-edge error rate is within a constant factor of optimal for polynomial-query black-box simulations (Theorem~\ref{thm:range-negative-result}). 

\subsection{Prior work} \label{sec:prior}
Classical coding theory methods designed for data transmission cannot be efficiently applied on a per-round basis to interactive protocols: either the slow-down or the error probability will be large. This problem was first addressed by Schulman for the case of two-party interactions. \cite{sch92} treated stochastic (positive capacity) channels and constructed a randomized compiler that transforms any $T$-round two-party protocol into a computationally efficient $\O(T)$-round simulation protocol. Later~\cite{sch93,sch96} treated adversarial noise and constructed a deterministic compiler that transforms any $T$-round two-party protocol into an $\O(T)$-round simulation protocol which tolerates adversarial error at the constant bit error rate $\frac{1}{240}$. This simulation, however, was not computationally efficient against adversarial error; it also relies on tree codes, which were shown to exist but have not yet been constructed (but see~\cite{b12,ms14}).
 Since then, the original results have been improved in many respects. As mentioned above,~\cite{rs94} treated the multiparty case for stochastic errors; since this solution depended upon tree codes, subsequent work provided effective simulations for a restricted class of communication protocols~\cite{ors05,ors09}. Gelles, Moitra and Sahai~\cite{gms11,gms14} provided a computationally efficient simulation of multiparty protocols against stochastic errors which avoids the per-instance pre-sharing of random bits in~\cite{sch92}. Returning to the two-party problem, Braverman and Rao~\cite{bra11} improved the adversarial error rate to $\frac{1}{8}$. The simulations in~\cite{gms11,gms14} work even against adversarial errors for two parties, but they are no longer computationally efficient in that setting. Brakerski and Kalai~\cite{bk12} and Brakerski and Naor~\cite{bn13} constructed computationally efficient  simulations at constant adversarial error rates. Several papers focused on noise thresholds for various channels~\cite{ghs14,gh14,egh15,be14}, while~\cite{cpt13} investigated what is possible while preserving the privacy of information not released by the noiseless protocol. Haeupler~\cite{h14} showed how to extend the non-tree-code-based randomized protocol in~\cite{sch92} to cope with adversarial error and at high rate.
Kol and Raz~\cite{kr13} showed a strict separation between the communication rates in one-way and interactive two-party communication.

The paper closest to our work is~\cite{JKL15}, which initiated the study of adversarial noise in protocols among $n>2$ parties. The main points of comparison are: (a)
 We provide simulation protocols for general networks, not only those containing a spanning star subgraph---in this respect our work is more general. 
(b) We consider the edges of the network to be capable of carrying simultaneously one bit per edge per unit time, rather than there being only a single edge of the network on which active communication is occurring at any time---in this sense the two works are incomparable, the model in~\cite{JKL15} favoring communication complexity and ours favoring round complexity.

\subsection{Notation} 
All of our graphs will be simple (i.e., without loops or multiple edges).
For positive results it suffices to show how to simulate the \emph{universal protocol} $\pi^*[G, T]$, which is a $T$-round deterministic protocol running on the digraph $G$ defined as follows. For a party $P_i$ with indegree $d_i^-$ and outdegree $d_i^+$, a $T$-round \emph{transmission function} for $P_i$ is a function $x_i$, which takes as input $d_i^-$ equal-length sequences of $< T$ bits received and gives as output $d_i^+$ bits to transmit. In $\pi^*[G, T]$, each party receives a transmission function as input, does as it instructs, and gives as output all the bits that she received. We will just write $\pi^*$ if $G$ and $T$ are clear.

\section{The Noise Threshold for Adversaries with a Global Budget}
\subsection{Asymptotically optimal error tolerance,
and fast simulation, on undirected networks}  \label{sec:reroute}
It was already shown by~\cite{JKL15} that reliable communication in an undirected $n$-vertex network is impossible against an adversary who can modify $\O(1/n)$ of the bit transmissions. (The model in~\cite{JKL15} is different but their argument applies mutatis mutandis to ours.) Our contribution is the converse to this statement:

\begin{thm} \label{thm:sparsifier-error-correction}
There exists a compiler $C$ such that if $\pi$ is a $T$-round protocol on a connected, undirected graph, then $C(\pi)$ tolerates a bit error rate of $\Omega(\frac{1}{n})$ and has a round complexity of $\O\left(\frac{m \log n}{n}T\right)$.
\end{thm}

\subsubsection{The RS compiler}

The main coding-theoretic ingredient in the proof of Theorem~\ref{thm:sparsifier-error-correction} is (a slight variant of) the \emph{RS compiler}. The RS compiler was designed for stochastic errors, but it turns out to have good properties in the adversarial setting as well:
\begin{prop} \label{prop:rs94}
There exists a compiler $C$ (the RS compiler) such that if $\pi$ is a $T$-round protocol on a digraph $G$, then $C(\pi)$ tolerates a bit error rate of $\Omega(\frac{1}{m})$ and has a round complexity of $\O(T)$.
\end{prop}
Proposition~\ref{prop:rs94} follows easily from the analysis in \cite{rs94}. We defer  proof to Appendix~\ref{apx:rs}, where we prove a much stronger claim about the RS compiler, that is needed for adversaries with per-edge budgets (Theorem~\ref{thm:range-positive-result}).

Since the simulations output by the RS compiler tolerate an error rate of $\Omega(1 / m)$, we can increase error tolerance to $\Omega(1 / \tilde{m})$ by first rerouting messages through a subgraph with $|\tilde{E}| = \tilde{m}$ edges. (See Equation~\ref{eqn:reroute}.) Naturally, we incur some round complexity overhead when we reroute through a sparse subgraph; most of the effort in this section will go toward minimizing this overhead.

%\ljs{Theorem 1 needs to come up to the front of the section}
% WMH: resolved

\subsubsection{Sparsification}

For a weighted, undirected graph $(G, w)$, let $\mathcal{L}_G(w)$ denote its Laplacian matrix. We will use the following theorem by de Carli Silva, Harvey, and Sato, which builds on
% the construction by Batson, Spielman, and Srivastava 
\cite{bss09} (improving in turn on the earlier~\cite{bk96}).

\begin{lem}[{\cite[Corollary 5]{dcshs11}}] \label{lem:spectral-sparsifier}
Suppose $G = (V, E)$ is an undirected graph, $w : E \to \R_+$ is a weight function, and $E = E_1 \cup \dots \cup E_k$ is a partition of the edge set. For any real $\epsilon \in (0, 1)$, there is a deterministic polynomial-time algorithm to find a subgraph $\tilde{G} = (V, \tilde{E})$ of $G$ and a weight function $\tilde{w}: \tilde{E} \to \R_+$ such that
\begin{equation} \label{eqn:spectral-sparsifier}
x^T \mathcal{L}_G(w) x \leq x^T \mathcal{L}_{\tilde{G}}(\tilde{w}) x \leq (1 + \epsilon) x^T \mathcal{L}_G(w) x \quad \text{for all $x \in \R^n$,}
\end{equation}
\begin{equation} \label{eqn:cost-control}
\sum_{e \in E_i} w_e \leq \sum_{e \in \tilde{E} \cap E_i} \tilde{w}_e \leq (1 + \epsilon) \sum_{e \in E_i} w_e \quad \text{for all $1 \leq i \leq k$,}
\end{equation}
and $|\tilde{E}| \in \O\left(\frac{n + k}{\epsilon^2}\right)$. 
\end{lem}

The following is a straightforward consequence of Lemma~\ref{lem:spectral-sparsifier}.
\begin{lem} \label{lem:cut-sparsifier}
Suppose $G = (V, E)$ is a connected, undirected graph. There exists a subgraph $\tilde{G} = (V, \tilde{E})$ with $|\tilde{E}| \in \O(n)$ such that for every cut $U \subseteq V$,
\begin{equation}
\frac{5m}{n} \left|\tilde{\delta}(U)\right| \geq |\delta(U)|,
\end{equation}
where $\delta(U)$ is the set of edges in $G$ crossing $U$, and $\tilde{\delta}(U)$ is the set of edges in $\tilde{G}$ crossing $U$.
\end{lem}
\begin{proof}
Define $w(e) = 1$ for every $e \in E$. Partition the edge set $E$ into $n$ sets $E = E_1 \cup \dots \cup E_n$, where each $E_i$ has at most $\lceil \frac{m}{n} \rceil$ edges in it. Pick $\epsilon = \frac{1}{2}$, and let $\tilde{G}$ be as in Lemma~\ref{lem:spectral-sparsifier}. Consider an arbitrary $e \in \tilde{E}$, say with $e \in E_i$. By Equation~\ref{eqn:cost-control},
\begin{equation}
\sum_{e \in \tilde{E} \cap E_i} \tilde{w}_e \leq \frac{3}{2}\sum_{e \in E_i} w_e.
\end{equation}
Since $w_e = 1$ for all $e$, the right-hand side is just $\frac{3}{2}|E_i|$, which is $ \leq \frac{3}{2} \lceil m / n \rceil$. Thus, in particular, $\tilde{w}_e \leq \frac{3}{2} \lceil m / n \rceil$.
Now, consider an arbitrary cut $U \subseteq V$. Let $x \in \R^n$ be the indicator function for $U$. By Equation~\ref{eqn:spectral-sparsifier},
\begin{equation}
\sum_{e \in \delta(U)} w_e \leq \sum_{e \in \tilde{\delta}(U)} \tilde{w}_e.
\end{equation}
Since $w_e = 1$, the left sum is just $|\delta(U)|$. Since every $\tilde{w}_e \leq \frac{3}{2} \lceil m / n \rceil$, the right sum is $\leq \frac{3}{2} |\tilde{\delta}(U)| \cdot \lceil m / n \rceil$. Thus,
\begin{equation}
|\delta(U)| \leq \frac{3}{2} \left\lceil \frac{m}{n} \right\rceil \left|\tilde{\delta}(U)\right| \leq \frac{3}{2} \left(\frac{m}{n} + 1 \right) \left|\tilde{\delta}(U)\right| \leq \frac{5m}{n} \left|\tilde{\delta}(U)\right|.
\end{equation}
\end{proof}

\subsubsection{Routing}

Suppose $\mathcal{N}$ is a multicommodity flow network on $G = (V, E)$; let $d_i$ denote the demand of commodity $i$. We say that the \emph{value} of a flow $F$ is the largest number $\lambda \in [0, 1]$ such that for every $i$, $\lambda d_i$ units of commodity $i$ flow from the source of $i$ to the sink of $i$ in $F$. The \emph{maximum concurrent flow} of $\mathcal{N}$ is the largest value of any flow. For any cut $U \subseteq V$, we let $\capacity(U)$ denote the sum of the capacities of edges crossing $U$, and we let $\demand(U)$ denote the sum of the demands of commodities whose sources and sinks are on opposite sides of $U$. We rely on the following approximate max-flow min-cut theorem for multicommodity flow in undirected networks, due to Linial, London, and Rabinovich.

\begin{lem}[{\cite[Theorem 4.1]{llr95}}] \label{lem:multicommodity-max-flow-min-cut}
Let $\mathcal{N}$ be a $k$-commodity undirected flow network on $G = (V, E)$ and 
 $\lambda$ its maximum concurrent flow. 
 There is a deterministic polynomial-time algorithm which, given $\mathcal{N}$, finds a cut $U \subseteq V$ such that
\begin{equation}
\frac{\capacity(U)}{\demand(U)} \leq \O(\log k) \cdot \lambda.
\end{equation}
% where $\lambda$ is the maximum concurrent flow of $\mathcal{N}$. 
\end{lem}

Given just a digraph $G$, we can naturally define an $m$-commodity flow network $\mathcal{N}_G$ on $G$: the commodity associated with edge $(P_i, P_j)$ has source $P_i$, sink $P_j$, and demand $1$; every edge has capacity $1$. Combining Lemma~\ref{lem:cut-sparsifier} with Lemma~\ref{lem:multicommodity-max-flow-min-cut}, we can prove the following lemma.

\begin{lem} \label{lem:sparse-flow}
Suppose $G$ is an undirected graph. There exists a flow for $\mathcal{N}_G$ with value $\Omega\left(\frac{n}{m \log m}\right)$ which uses only $\O(n)$ edges.
\end{lem}

\begin{proof}
Let $\tilde{G}$ be as in Lemma~\ref{lem:cut-sparsifier}, and let $\mathcal{N}$ denote the $m$-commodity flow network on $\tilde{G}$ with all the same commodities as $\mathcal{N}_G$ (and with every edge in $\tilde{G}$ still having capacity $1$.) For any cut $U$, the capacity $\capacity(U)$ is just the number of edges in $\tilde{G}$ which cross $U$, i.e. $|\tilde{\delta}(U)|$; the demand $\demand(U)$ is just the number of edges in $G$ which cross $U$, i.e. $|\delta(U)|$. Thus, if we let $U$ be that guaranteed by Lemma~\ref{lem:multicommodity-max-flow-min-cut} for $\mathcal{N}$, we have
\begin{equation}
\frac{n}{5m} \leq \frac{|\tilde{\delta}(U)|}{|\delta(U)|} = \frac{\capacity(U)}{\demand(U)} \leq \O(\log m) \cdot \lambda,
\end{equation}
and hence $\lambda \in \Omega(\frac{n}{m \log m})$. Of course, the same flow which achieves this $\lambda$ in $\mathcal{N}$ can be used in $\mathcal{N}_G$, completing the proof.
\end{proof}

Flows are allowed to be fractional, but ultimately, we are interested in integer flows (i.e. collections of paths.) The following lemma quantifies the sense in which fractional flows do not cause too much trouble.

\begin{lem} \label{lem:flow-to-paths}
Suppose $G = (V, E)$ is a digraph, and there is a flow $F$ for $\mathcal{N}_G$ with value $\lambda$ which only uses $s$ edges. Then there exists a set $\mathcal{P}$ of $m$ paths through $G$, containing one path from $P_i$ to $P_j$ for each $(P_i, P_j) \in E$, which uses at most $s$ distinct edges in total and which has congestion at most $9(\frac{1}{\lambda} + \ln m)$.
\end{lem}

The proof of Lemma~\ref{lem:flow-to-paths} is a straightforward probabilistic argument, which we defer to Appendix~\ref{apx:flow-to-paths}.

\begin{lem} \label{lem:useful-paths}
Suppose $G = (V, E)$ is a connected, undirected graph. There exists a subgraph $\tilde{G} = (V, \tilde{E})$ with $\O(n)$ edges and a set $\mathcal{P}$ of $m$ simple paths through $\tilde{G}$, such that
\begin{enumerate}[(i)]
\item $\mathcal{P}$ contains one path from $P_i$ to $P_j$ for each $(P_i, P_j) \in E$, and
\item $\mathcal{P}$ has dilation $\O(\frac{m \log n}{n})$ and congestion $\O(\frac{m \log n}{n})$.
\end{enumerate} 
\end{lem}

\begin{proof}
From Lemmas~\ref{lem:sparse-flow} and \ref{lem:flow-to-paths}, there exists a set $\mathcal{P}_0$ of $m$ paths, containing one path from $P_i$ to $P_j$ for each $(P_i, P_j) \in E$, which uses $\O(n)$ distinct edges in total and which has congestion $\O(\frac{m \log m}{n})$. Let $p_{ij}$ denote the path from $P_i$ to $P_j$ in $\mathcal{P}_0$. Define a path $p_{ij}'$ from $P_i$ to $P_j$ by
\begin{equation}
p'_{ij} =
\begin{cases}
p_{ij} & \text{if $p_{ij}$ has length $\leq \frac{m \log m}{n}$} \\
(P_i, P_j) & \text{otherwise.}
\end{cases}
\end{equation}
Let $\mathcal{P} = \{p'_{ij} : (P_i, P_j) \in E\}$, and let $\tilde{E}$ be the set of edges used by $\mathcal{P}$. Because of the bounds of $\mathcal{P}_0$, the sum of the lengths of the paths in $\mathcal{P}_0$ must be $\O(m \log m)$. Therefore, in particular, the number of paths in $\mathcal{P}_0$ of length at least $\frac{m \log m}{n}$ is $\O(n)$. Therefore, $\mathcal{P}$ still uses only $\O(n)$ distinct edges in total. Furthermore, by construction, the dilation of $\mathcal{P}$ is no more than $\frac{m \log m}{n}$. Finally, the congestion on an edge $e$ in $\mathcal{P}$ is no more than the congestion of that edge in $\mathcal{P}_0$, plus $1$ for the length-$1$ path across $e$ which may be in $\mathcal{P} \setminus \mathcal{P}_0$. Thus, in particular, the congestion of $\mathcal{P}$ is still $\O\left(\frac{m \log m}{n}\right)$. Of course, $\O(\log m) = \O(\log n)$, so we are done.
\end{proof}

\subsubsection{Scheduling}

We use the following fundamental theorem of Leighton, Maggs and Rao.

\begin{lem}[{\cite[Theorem 3.4]{lmr94}}] \label{lem:magic-scheduling-theorem}
Suppose $G$ is a digraph, and $\mathcal{P}$ is a set of simple paths through $G$ with dilation $\ell$ and congestion $c$. There exists a schedule for routing packets along the paths in $\mathcal{P}$, with at most one packet traversing each edge in each time step, in a total of $\O(c + \ell)$ time steps. 
\end{lem}

\begin{proof}[of Theorem~$\ref{thm:sparsifier-error-correction}$]
It suffices to describe $\tilde{\pi}^* = C(\pi^*)$. The compiler $C$ is formed by composing a ``sparsifying compiler'' with the RS compiler, as depicted in Equation~\ref{eqn:reroute}.
\begin{equation} \label{eqn:reroute}
C: \pi^* \quad \stackrel{\text{Sparsifying compiler}}{\longmapsto} \quad \pi' \quad \stackrel{\text{RS compiler}}{\longmapsto} \quad \tilde{\pi}^*
\end{equation}
Let $\tilde{G}$ and $\mathcal{P}$ be as in Lemma~\ref{lem:useful-paths}; the intermediate protocol $\pi'$ runs on $\tilde{G}$. On input $x = (x_1, \dots, x_n)$, each round of $\pi^*$ is simulated by $\O\left(\frac{m \log n}{n}\right)$ rounds in $\pi'$ as follows. Assume inductively that we have already simulated $\tau$ rounds. Based on these simulations, for each $(P_i, P_j) \in E$, there is some bit $b_{ij}$ which $x_i$ instructs $P_i$ to send to $P_j$ during round $\tau + 1$ of $\pi^*$. By Lemma~\ref{lem:magic-scheduling-theorem}, there is a schedule by which the parties can coordinate so that every $b_{ij}$ reaches its destination after $\O\left(\frac{m \log n}{n}\right)$ rounds; the parties follow this schedule. Thus, $\pi'$ successfully simulates $\pi^*$ on a noiseless network, and runs in $\O\left(\frac{m \log n}{n}T\right)$ rounds. Therefore, by Proposition~\ref{prop:rs94}, $\tilde{\pi}^*$ tolerates an error rate of $\Omega(\frac{1}{n})$ as a simulation of $\pi^*$, and still runs in $\O\left(\frac{m \log n}{n}T\right)$ rounds.
\end{proof}

%We remark that the only barrier to computational efficiency in this simulation is that of efficiently constructing tree codes (which are used by the RS compiler.) 
The sparse subgraph can be efficiently constructed, as stated in Lemma~\ref{lem:spectral-sparsifier}. There are efficient algorithms for constructing multicommodity flows that are within a factor of $1 + \epsilon$ of optimal; see e.g. \cite{mad10}. The proofs of Lemmas~\ref{lem:flow-to-paths} and~\ref{lem:useful-paths} can be implemented as efficient randomized algorithms in a straightforward way. Efficient randomized algorithms are also known which construct schedulers with the parameters of Lemma~\ref{lem:magic-scheduling-theorem}~\cite{lmr99}. The RS compiler is not computationally efficient in the presence of adversarial errors.

\subsection{The noise threshold in arbitrary digraphs}
\subsubsection{Positive result (lower bound on tolerable error rates)}
We can now also easily obtain a lower bound on the maximum tolerable error rate on arbitrary \emph{directed} graphs; in this setting, results on undirected sparsification do not help us to reduce the round complexity of the simulation. What is most interesting here is identification of the graph parameter that governs the adversarial noise threshold.

\begin{thm} \label{thm:reachability-preserving-subgraph-error-rate}
Suppose $G = (V, E)$ is a digraph without isolated vertices, and suppose each minimum equivalent digraph of $G$ has $s$ edges. There exists a compiler $C$ such that if $\pi$ is a $T$-round protocol on $G$, then $C(\pi)$ tolerates a bit error rate of $\Omega(\frac{1}{s})$ and has round complexity $\O(mT)$.
\end{thm}

\begin{proof}
Pick some minimum equivalent digraph $\tilde{G} = (V, \tilde{E})$. Define $\mathcal{P}$ to include, for each $(P_i, P_j) \in E$, some simple path from $P_i$ to $P_j$ through $\tilde{G}$. Clearly, $\mathcal{P}$ has dilation no more than $n$ and congestion no more than $m$, and there are at most $s$ distinct edges used by $\mathcal{P}$. The same construction as in the proof of Theorem~\ref{thm:sparsifier-error-correction} works here.
\end{proof}

We remark that finding a minimum equivalent digraph is NP-hard, but there is a polynomial-time approximation algorithm with a performance guarantee of about 1.64 \cite{kry02}.

\subsubsection{Negative result (upper bound on tolerable error rates)} \label{sec:negative-results}

%It is not hard to see that the error rate in Theorem~\ref{thm:sparsifier-error-correction} is within a constant factor of optimal: as noted by~\cite{JKL15}, with a bit error rate budget of $\frac{1}{n}$, the adversary can afford to effectively cut one party off from the rest of the graph. 
We now show
% more generally 
that the error rate of Theorem~\ref{thm:reachability-preserving-subgraph-error-rate} is within a constant factor of optimal. We begin with the following result by Moyles and Thompson.

\begin{lem}[{\cite[Theorem 1]{mt69}}] \label{lem:acyclic-subgraph}
Suppose $G = (V, E)$ is a directed acyclic graph, and $\tilde{G} = (V, \tilde{E})$ is a minimum equivalent digraph of $G$. Then $\tilde{E}$ is exactly the set of edges $(P_i, P_j) \in E$ such that there is no path from $P_i$ to $P_j$ through $G$ which avoids the edge $(P_i, P_j)$. \qed
\end{lem}

Suppose $G = (V, E)$ is a digraph.
We define the \emph{relative edge connectivity (REC)} of $G$ to be the least $k$ such that there are $k$ edges whose removal from $G$ changes the reachability relation. For example, if $G$ is strongly connected, then its REC is simply its edge connectivity.

\begin{lem} \label{lem:reachability-preserving-subgraph-relative-edge-connectivity}
Suppose $G = (V, E)$ is a digraph with no isolated vertices, with $\rec(G)=k$. Then it has a reachability-equivalent subgraph $\tilde{G} = (V, \tilde{E})$ with no more than $5m / k$ edges.
\end{lem}

\begin{proof}
Let $G_1, \dots, G_q$ be the strongly connected components of $G$. Let $V^* = \{G_i\}_i$, and let $G^* = (V^*, E^*)$ be the condensation of $G$. Define a weight function $w: E^* \to \N$ by saying that the weight of $(G_i, G_j)$ is the number of edges in $E$ going from $G_i$ to $G_j$. By Lemma~\ref{lem:acyclic-subgraph}, there is a reachability-equivalent subgraph $\tilde{G}^* = (V^*, \tilde{E}^*)$ of $G^*$, such that for each $(G_i, G_j) \in \tilde{E}^*$, every path from $G_i$ to $G_j$ through $G^*$ uses the edge $(G_i, G_j) \in E^*$. Therefore, each edge $(G_i, G_j) \in \tilde{E}^*$ must have weight at least $k$, since removing the edges from $G_i$ to $G_j$ in $G$ would make the vertices in $G_j$ unreachable from the vertices in $G_i$.

We form the subgraph $\tilde{G} = (V, \tilde{E})$ as follows. For each $G_i$, we define $\tilde{E}_i$ to be the set of edges in a minimum equivalent digraph of $G_i$. We define $\tilde{E}^D$ to contain one edge from $G_i$ to $G_j$ for each $(G_i, G_j) \in \tilde{E}^*$. We define $\tilde{E} = \tilde{E}_1 \cup \dots \cup \tilde{E}_q \cup \tilde{E}^D$. By construction, clearly, $\tilde{G}$ is a reachability-equivalent subgraph of $G$.

Say each $G_i$ has $n_i$ vertices. Then $\tilde{E}_i$ has no more than $2(n_i - 1)$ edges, because we can form a reachability-equivalent subgraph of $G_i$ with $2(n_i - 1)$ edges by picking a root vertex $P_i$ in $G_i$ and including all edges in an in-branching of $G_i$ rooted at $P_i$, as well as all edges in an out-branching of $G_i$ rooted at $P_i$. Therefore, the $\tilde{E}_i$s have, in total, no more than $2n$ edges. Furthermore, $\tilde{E}^D$ has no more than $m/k$ edges, since each edge in $\tilde{E}^*$ has weight $k$. Now, $n \leq 2m / k$, because $k$ is no more than the minimum number of edges adjacent to any vertex. Therefore, in total, $\tilde{G}$ has no more than $4m / k + m / k = 5m / k$ edges.
\end{proof}

\begin{lem} \label{lem:relative-edge-connectivity-error-rate}
Suppose $C$ is a compiler and $G = (V, E)$ is a digraph. Suppose that for some $T > 0$, $C(\pi^*[G, T])$ runs on a graph $\tilde{G} = (V, \tilde{E})$ with $\tilde{m} = |\tilde{E}|$ edges. Define $\lambda$ to be $\rec\left(\tilde{G}\right)$ if $\tilde{G}$ is reachability-equivalent to $G$, and $\lambda = 0$ otherwise. Then the failure probability of $C(\pi^*[G, T])$ in the presence of the bit error rate $\lambda / \tilde{m}$ is at least $1 - 2^{-T}$.
\end{lem}

\begin{proof}
Say $(P_i, P_j) \in E$ and $S$ is a set of $\lambda$ edges in $\tilde{E}$ such that after removing all the edges in $S$ from $\tilde{G}$, there is no path from $P_i$ to $P_j$. Consider the adversary $\mathcal{A}$ who zeroes out all messages sent across every edge $e \in S$. Consider choosing an input $x$ uniformly at random. For any transcript at $P_j$, the probability of that transcript conditioned on any input that $P_i$ might receive is equally likely. Thus, $P_j$ has only a $2^{-T}$ chance of correctly guessing the $T$ bits that $P_i$ would have sent $P_j$ if they had followed $\pi^*[G, T]$.
\end{proof}

Observe that Lemma~\ref{lem:relative-edge-connectivity-error-rate} shows that on undirected networks with bounded edge connectivity, the error rate $\Omega(\frac{1}{m})$ is optimal, among simulations which run on that same network.

\begin{thm} \label{thm:error-rate-negative-result}
Suppose $C$ is a compiler and $G$ is a digraph without isolated vertices, for which each minimum equivalent digraph has $s$ edges. Then for all $T > 0$, the failure probability of $C(\pi^*[G, T])$ in the presence of the bit error rate $5 / s$ is at least $1 - 2^{-T}$.
\end{thm}

\begin{proof}
Let $\tilde{G} = (V, \tilde{E})$ be the graph on which $C(\pi^*[G, T])$ runs. Say $\tilde{k} = \rec\left(\tilde{G}\right)$, and $\tilde{m} = \left|\tilde{E}\right|$. If $\tilde{G}$ is not reachability-equivalent to $G$, we are done by Lemma~\ref{lem:relative-edge-connectivity-error-rate}. Otherwise, $s \leq \tilde{s}$, where $\tilde{s}$ is the number of edges in each minimum equivalent digraph of $\tilde{G}$. By Lemma~\ref{lem:reachability-preserving-subgraph-relative-edge-connectivity}, $\tilde{s} \leq 5 \tilde{m} / \tilde{k}$. Thus, $5 / s \geq \tilde{k} / \tilde{m}$; an application of Lemma~\ref{lem:relative-edge-connectivity-error-rate} completes the proof.
\end{proof}

\subsection{Lower bound on the round complexity of robust simulations}

If $G$ has a small relative edge connectivity, then simulations of protocols on $G$ must run on subgraphs to achieve optimal error tolerance. Naturally, there is a round complexity cost associated with moving to a sparse subgraph. These two ideas prove the following theorem.

\begin{thm} \label{thm:runtime-lower-bound}
Suppose $C$ is a compiler, $\rho \in [0, 1]$, and $G$ is a digraph with $\rec(G) = k$. Suppose the round complexity of $C(\pi^*[G, T])$ is less than $\frac{m \rho}{k} T$. Then the failure probability of $C(\pi^*[G, T])$ in the presence of the bit error rate $\rho$ is at least $\frac{1}{2}$.
\end{thm}

\begin{proof}
Suppose $C(\pi^*[G, T])$ runs on a subgraph $\tilde{G}$ of $G$, with $\tilde{m}$ edges and with $\rec\left(\tilde{G}\right) = \tilde{k}$. By Lemma~\ref{lem:relative-edge-connectivity-error-rate}, if $\rho \geq \tilde{k} / \tilde{m}$, we are done, so assume $\rho < \tilde{k} / \tilde{m}$. We are also done if $\tilde{G}$ is not reachability-equivalent to $G$, so assume that it is, which implies that $\tilde{k} \leq k$, and hence $\tilde{m} < k / \rho$.

Say $\tilde{\pi}^*[G, T]$ runs in $\tilde{T}$ rounds, with $\tilde{T} < Tm/\tilde{m}$. Observe that the average indegree in $\tilde{G}$ is no more than $\tilde{m} / m$ times the average indegree in $G$, so there is some party $P_i$ whose indegree $\tilde{d}_i^-$ in $\tilde{G}$ is no more than $\tilde{m} / m$ times her indegree $d_i^-$ in $G$. Fix some input $x_i$ for $P_i$, and choose every other party's input uniformly at random. At the end of the execution of the simulation protocol, $P_i$ must guess $d_i^- T$ bits based on $\tilde{d}_i^- \tilde{T}$ bits that she receives. Since $\tilde{d}_i^- \tilde{T} < d_i T$, the probability of success is no more than $\frac{1}{2}$.
\end{proof}

When $G$ is connected and undirected, taking $\rho \in \Omega(\frac{1}{n})$ in Theorem~\ref{thm:runtime-lower-bound} shows that the round complexity blowup of Theorem~\ref{thm:sparsifier-error-correction} is within a factor of $\O(k \log n)$ of optimal, where $k$ is now just the edge connectivity of $G$. On highly connected graphs, this leaves a sizable gap. It is quite possible that there are compilers with optimal error tolerance and with round complexity lower than that achieved in Theorem~\ref{thm:sparsifier-error-correction}. The following theorem establishes that this is at least true if the parties share access to a common random string (unavailable to the adversary), and if we make a strong assumption on connectivity.
\suppress{
As some evidence, we give the following theorem, which applies if we make the extra assumption that the parties have access to a common random string, to which the adversary does not have access.} (We defer the proof to Appendix~\ref{apx:magi}.)

\begin{thm} \label{thm:magi}
Suppose $G = (V, E)$ is an undirected graph such that for every $(P_i, P_j) \in E$, the endpoints $P_i$ and $P_j$ have $\Omega(n)$ common neighbors. There exists a shared-randomness compiler $C$ such that if $\pi$ is a $T$-round protocol on $G$, then $C(\pi)$ tolerates a bit error rate of $\Omega(\frac{1}{n})$ with failure probability $e^{-\Omega(T)}$, and $C(\pi)$ has round complexity $\O(T \log n)$.
\end{thm}

\section{The Noise Threshold for Adversaries with a Per-Edge Budget} \label{sec:edge-error-rate}
In this section, we are interested in \emph{per-edge error rates}, where we restrict the distribution of errors as well as the total number. Specifically, we say that an adversary stays within the per-edge error rate $\rho$ budget if on each edge, the fraction of bits transmitted on that edge which are flipped is no more than $\rho$. Note that when we are considering per-edge error rates, we can assume without loss of generality that simulations run on the same graphs as the original protocols.

For a digraph $G$, we define the \emph{signal diameter} $D$ of $G$ to be the maximum finite distance between any two vertices in $G$. That is, the signal diameter of $G$ is the maximum, over all $P_i, P_j$ for which $P_j$ is reachable from $P_i$, of the length of the shortest path from $P_i$ to $P_j$. For example, if $G$ is strongly connected, then the signal diameter of $G$ is just the ordinary diameter of $G$.

\subsection{Optimal per-edge error rates, ignoring round complexity} 
\paragraph{Positive result:}

\begin{thm} \label{thm:signal-diameter-positive-result}
Suppose $G$ is a digraph with signal diameter $D$. For every $\epsilon > 0$, there exists a compiler $C$ such that if $\pi$ is a protocol on $G$, then $C(\pi)$ tolerates the per-edge error rate $\frac{1}{4D} - \epsilon$, and $C(\pi)$ has round complexity $\O(D^2 nT 2^{nT})$.
\end{thm}

\begin{proof}
We describe $C(\pi^*[G, T])$. By the Gilbert-Varshamov bound, there is some family of error correcting codes $\chi$ with positive asymptotic rate and minimum relative distance at least $\frac{1}{2} - D \epsilon$. Let $\ell$ be sufficiently long so that for any length-$D$ list $F$ of $T$-round transmission functions on $G$, $\chi(F)$ has length no more than $\ell$. Set $\tilde{T} = D \ell$. Divide the $\tilde{T}$ rounds into $D$ segments of length $\ell$. In the $j$th segment, $P_i$ transmits (to all of her out-neighbors) the encoding under $\chi$ of the list $(x_{k_1}, \dots, x_{k_m})$ of transmission functions of parties $P_{k_s}$ such that there is a path of length $< j$ from $P_{k_s}$ to $P_i$.

First, suppose some decoding operation failed. Because of the minimum relative distance property, the adversary must have introduced at least $\frac{1}{4} \ell (1 - D\epsilon)$ bit errors on some edge, which is a per-edge error rate of $\frac{1}{4D} - \frac{1}{4}\epsilon$, which exceeds the specified budget. Suppose instead that all decoding operations succeed. Then every party $P_i$ knows $x_j$ for every party $P_j$ from which $P_i$ is reachable. Using this information, $P_i$ can infer all of the bits that she would have received if the parties had followed $\pi^*$ on a noiseless network. Finally, for the round complexity estimate, note that trivially every party has degree at most $n$. Hence, to specify a $T$-round transmission function, it suffices to specify the $nT$ bits that a party would send, given any arbitrary length-$n$ list of $T$-bit incoming strings. Hence, a list of $D$ such transmission functions can be specified with $D nT 2^{nT}$ bits, so $\ell$ is $\O(Dn T 2^{nT})$.
\end{proof}

\paragraph{Negative result:}
We give a matching (up to a factor of $2$) negative result, showing that the error rate $\frac{1}{2D}$ cannot be tolerated.

\begin{thm} \label{thm:signal-diameter-negative-result}
Suppose $C$ is a compiler and $G$ is a digraph with signal diameter $D$. Then for all sufficiently large $T$, there exists a $T$-round protocol $\pi$ such that the failure probability of $C(\pi)$ in the presence of the per-edge error rate $\frac{1}{2D}$ is at least $\frac{1}{4}$.
\end{thm}
\begin{proof}
Select vertices $P_0,\ldots,P_D$ such that $(P_0,\ldots,P_D)$ is a shortest path from $P_0$ to $P_D$. Pick any integer $L > 2mD$. In the protocol $\pi_L$, $P_0$ receives an $L$-bit string $x$ as input, and transmits it to $P_D$ along a shortest path, so that $\pi_L$ runs in $T = L + D - 1$ rounds. Say $C(\pi_L)$ runs in $\tilde{T}$ rounds. %\phantom\qedhere

The strategy of the adversary is to sample two possible inputs $x,x'$ to $P_0$ in such a way that (a) $x \neq x'$ with probability at least $1/2$; (b) The probability distribution on the transcripts of all channel communications leading into $P_D$, is the same whether $x$ or $x'$ were given to $P_0$ as input. The theorem will follow.

However, the adversary cannot commit to the pair $x,x'$ at the very outset so her strategy is slightly more complicated. Let $\ell$ be the largest even integer s.t.~$\ell \leq \tilde{T} / D$, and let $B = \tilde{T} - D\ell$. Note then that $B < 2D$.

First the adversary selects an input $x$ u.a.r.\ in $\{0,1\}^L$. Then she allows the protocol to proceed without interference for $B$ rounds. 
Let $\chi \in \{0, 1\}^{mB}$ be the random variable denoting the transcript generated by all parties in the network during these rounds. The adversary knows the (possibly randomized) simulation protocol and therefore knows the conditional probabilities of transcripts given inputs. 
She now samples $x' \in \{0,1\}^L$ from the posterior distribution (given $\chi$ and the uniform prior on $\{0,1\}^L$). For $a,b\in \{0,1\}^L$ and $c \in \{0, 1\}^{mB}$, let $[a,b,c]$ denote the event that $a$ was chosen as the input $x$, $c$ was the transcript $\chi$, and $b$ was chosen as the ``alternate'' input $x'$. The key property of this construction is that for any $a,b,c$, $\Pr([a,b,c])=\Pr([b,a,c])$. 
%The key property of this construction is that for any $a,b,c$, $\Pr(x=a,x'=b,\chi=c)=\Pr(x'=a,x=b,\chi=c)$. 

Before continuing to describe the adversary's strategy, let us argue already why (a) holds. Consider using the following alternate sampling rule for $x'$: for each $\chi$, instead of selecting $x'$ from the posteriori distribution, select $x'$ to be the max-likelihood decoding of $\chi$. This can only increase $\Pr(x=x')$. This creates a deterministic decoding map from transcripts $\chi$ to inputs, which means that there is a set of at most $2^{mB}$ inputs $x$ on which it can ever occur that $x=x'$. The probability that $x$ is selected from this set is $2^{mB-L}<2^{2mD-L}\leq 1/2$. 

The adversary now breaks the remaining $D\ell$ rounds of the protocol into $D$ segments, each of $\ell$ rounds. Each segment is further broken into two half-segments, each of $\ell/2$ rounds.
Let $d(P,P')$ be the length of a shortest directed path (possibly infinite) from vertex $P$ to vertex $P'$. Let $V_k=\{P:d(P_0,P)=k\}$, and let $W_k=\bigcup_{k' \geq k} V_{k'}$. At the beginning of segment $k$ ($1\leq k \leq D$) she flips a fair coin to decide whether to attack the first or second half of the segment. During the half-segment that she attacks, she substitutes messages of her choice for all the messages from $V_{k-1}$ to $V_k$. The manner in which she generates these messages is as follows. 

The adversary's strategy is to simulate an imaginary, ``alternative reality'' portion of the network, that gradually grows. See Figures~\ref{fig:signal-diameter-negative-result-segment-3},~\ref{fig:signal-diameter-negative-result-segment-4}, which for space reasons are in Appendix~\ref{apx:signal-diameter-figs}. For the duration of the first segment, the simulated network portion consists of a single vertex $\bar{P}_0$, mirroring the actual network vertex $P_0$. During the second segment the simulated region grows to mirror the induced network on $\{P_0\} \cup V_1$ (or what is the same, $V_0\cup V_1$). In general during the $k$th segment the simulated network region is a copy of the induced graph on $V_0 \cup \ldots \cup V_k$. Throughout the entire protocol, the adversary continues simulating communications on this gradually growing region;  during attacking half-segments, the adversary replaces the $V_{k-1} \to V_k$ messages by $\bar{V}_{k-1}\to V_k$ messages, that is, she substitutes the outgoing messages of the simulated reality on $V_0 \cup \ldots \cup V_k$ for the outgoing messages of the real vertices in that region.

It remains to describe how the states of these imaginary vertices are initialized and updated.

Updates during a segment are as follows: during segment $k$, the state of each imaginary vertex in $\bar{V}_0 \cup \ldots \cup \bar{V}_{k-1}$ is updated in each round just as it would in the protocol, using its prior state and, as inputs, the communications from the other imaginary vertices together with any communications coming from real vertices in $\bar{V}_k$. 

Initialization at the beginnings of segments are as follows: at time $B$ (the beginning of the first segment), $\bar{P}_0$ is initialized with a random state $s$ chosen from the posteriori distribution conditional on her input being $x'$ and on all messages that her genuine counterpart $P_0$ sent and received through time $B$. For $k\geq 2$, at the beginning of segment $k$ (i.e., at time $B+(k-1)\ell$), we have to enlarge the simulation to include new vertices $\bar{V}_{k-1}$. The existing vertices (those in $\bar{V}_0\cup \ldots \cup \bar{V}_{k-2}$) continue from their current state. Each vertex $\bar{P}\in \bar{V}_{k-1}$ is initialized with a random state $s$ chosen from the posteriori distribution conditional on all messages that its genuine counterpart $P$ sent and received up through time $B+(k-1)\ell$. 

Notice that the simulation is evolved forward in each round whether or not this is an attacking round. The imaginary vertices are always responding to messages coming from amongst themselves and from the real vertices. All that changes is whether $V_k$ is hearing messages from $V_{k-1}$ or from $\bar{V}_{k-1}$.

The key claim is this. Let $\vec{s}$ denote the transcript at time $B+(k-1)\ell$ of \textit{all} messages ever received at vertices in $W_k$. Then:

\begin{lem} For all $\vec{s}$, $\Pr(\vec{s}\; | \; [x,x',\chi])=\Pr(\vec{s}\; | \; [x',x,\chi])$. \end{lem}
That is, to the vertices in $W_k$, the probability distribution over what they have (collectively) heard up until this time is the same whether the input is $x$ or $x'$.
\begin{proof} The proof is by induction on $k$. The base case is $k=1$ and is simply our initial condition that $\Pr(\chi \; | \; x)=\Pr(\chi \; | \; x')$. Now for $k\geq 2$, let us denote by $h=1$ ($h=2$) the event that the adversary attacks during the first (resp.\ second) half of the $(k-1)$'st segment. We claim:

(1) For all $\vec{s}$, $\Pr(\vec{s} \; | \; [x,x',\chi,h=1])=\Pr(\vec{s} \; | \; [x',x,\chi,h=2])$.

(2) For all $\vec{s}$, $\Pr(\vec{s} \; | \; [x,x',\chi,h=2])=\Pr(\vec{s} \; | \; [x',x,\chi,h=1])$.

We argue (1) (and (2) follows analogously). At the beginning of the $(k-1)$'st segment the claim was true by induction; we need to argue that it remains so at the end of the $(k-1)$'st segment, and this could break down only due to a difference in the statistics on messages from $V_{k-1}\to V_k$. This does not occur because for both events $[x,x',\chi,h=1]$ and $[x',x,\chi,h=2]$, what $W_k$ hears during the first half of the $(k-1)$'st segment, is messages from vertices ``in the $x'$ world"---more formally, in the event $[x,x',\chi,h=1]$ it is vertices in $\bar{V}_{k-1}$ acting as if the input is $x'$, while in the event $[x',x,\chi,h=2]$, it is vertices in $V_{k-1}$, with the true input being $x'$; while what $W_k$ hears during the second half of the $(k-1)$'st segment, is messages from vertices ``in the $x$ world"---more formally, in the event $[x,x',\chi,h=1]$ it is vertices in $V_{k-1}$, with the true input being $x$, while in the event $[x',x,\chi,h=2]$, it is vertices in $\bar{V}_{k-1}$ acting as if the input is $x$.

Finally, \begin{align*}
\Pr(\vec{s}\; | \; [x,x',\chi])&=\frac12 \Pr(\vec{s} \; | \; [x,x',\chi,h=1]) + \frac12 \Pr(\vec{s} \; | \; [x,x',\chi,h=2]) \\
&=
\frac12 \Pr(\vec{s} \; | \; [x',x,\chi,h=2]) +
\frac12 \Pr(\vec{s} \; | \; [x',x,\chi,h=1]) \\ &= \Pr(\vec{s}\; | \; [x',x,\chi]). %\qedhere
\end{align*}
\end{proof}
\end{proof}

\subsection{Optimal per-edge error rates for black-box simulations with polynomial query complexity} \label{sec:edge-error-rate-queries}

The proof of Theorem~\ref{thm:signal-diameter-positive-result} does not provide a useful compiler, since the round complexity $\tilde{T}$ blows up exponentially. In this section, we determine (up to a constant) the maximum tolerable per-edge error rate for polynomial-query compilers in a certain black-box model (described below).

For a digraph $G$, we define the \emph{chain-length} $R$ of $G$ to be the maximum, over all directed walks $W$ through $G$, of the number of distinct vertices visited in $W$. Observe that in any graph with at least one edge, the chain-length is strictly larger than the signal diameter, and that for a strongly connected graph, $R = n$.

\paragraph{Positive result:}

\begin{thm} \label{thm:range-positive-result}
Suppose $G$ is a digraph with chain-length $R$. There exists a compiler $C$ such that if $\pi$ is a $T$-round protocol on $G$, then $C(\pi)$ tolerates a per-edge error rate of $\Omega(\frac{1}{R})$ and has round complexity $\O(mT)$.
\end{thm}

Most of the effort required to prove Theorem~\ref{thm:range-positive-result} consists of a new analysis of the RS compiler. The key fact (whose proof we defer to Appendix~\ref{apx:rs}) is the following.

\begin{lem} \label{lem:rs-failure-walk}
There exists a compiler $C$ (the RS compiler) such that if $\pi$ is a $T$-round deterministic protocol on a digraph $G$ and an execution of $C(\pi)$ fails, then there is some walk through $G$ on the edges of which were at least $\frac{T}{48}$ bit errors. The round complexity of $C(\pi)$ is $\O(T)$.
\end{lem}

\begin{proof}[of Theorem~$\ref{thm:range-positive-result}$]
We let $C$ be as in the proof of Theorem~\ref{thm:reachability-preserving-subgraph-error-rate}, i.e. we reroute messages through a minimum equivalent digraph before using the RS compiler. Let $T' \in \O(m)T$ denote the round complexity of the intermediate protocol $\pi'$, so that the round complexity of $C(\pi^*)$ is $\O(T')$. If $C(\pi^*)$ fails, then by Lemma~\ref{lem:rs-failure-walk}, there is some walk $W$ through the minimum equivalent digraph $\tilde{G}$, on the edges of which were $\frac{T'}{48}$ bit errors. Since each strongly connected component $H$ of $\tilde{G}$ with $n'$ vertices has no more than $2(n' - 1)$ edges, the number of edges in $W$ is no more than $3R$. Thus, on some edge in $W$, there were $\frac{T'}{3R \cdot 48}$ bit errors, which is a per-edge error rate of $\Omega\left(\frac{1}{R}\right)$.
\end{proof}

\paragraph{Black-box negative result:} %\label{sec:black-box-negative-result}

We now give a result which shows that the per-edge error rate in Theorem~\ref{thm:range-positive-result} is within a constant factor of optimal, among polynomial-query compilers in a certain black-box model. Recall that in $\pi^*$ (or any simulation thereof), each party $P_i$ receives as input a transmission function $x_i$. We will call a simulation $\tilde{\pi}^*$ of $\pi^*$ a \emph{black-box simulation} if in $\tilde{\pi}^*$, the parties only ever access their inputs by making queries, wherein they specify an input to $x_i$ and are given the corresponding output. Naturally, the \emph{query complexity} of a black-box simulation is the largest number of total queries that the parties ever collectively make. A \emph{polynomial-query black-box compiler} is a compiler $C$ which takes as input a universal protocol $\pi^*[G, T]$ and gives as output a black-box simulation $C(\pi^*[G, T])$, such that for every graph $G$, there is a polynomial $\mathcal{Q}(T)$, so that the simulation $C(\pi^*[G, T])$ has a query complexity bounded by $\mathcal{Q}(T)$. Observe that (for a fixed graph $G$) the compiler which proved Theorem~\ref{thm:range-positive-result} makes $\O(T)$ queries per round, for a total query complexity of $\O(T^2)$. In contrast, the simulation in the proof of Theorem~\ref{thm:signal-diameter-positive-result} has exponential query complexity.

\begin{thm} \label{thm:range-negative-result}
Suppose $C$ is a polynomial-query black-box compiler. Then for any digraph $G$ with no isolated vertices and with chain-length $R$, the failure probability of $C(\pi^*[G, T])$ in the presence of the per-edge error rate $\frac{4}{R}$ goes to $1$ as $T \to \infty$.
\end{thm}

Both the statement and the proof of Theorem~\ref{thm:range-negative-result} are inspired by the black-box negative result in \cite{JKL15}. The proof is therefore deferred to Appendix~\ref{apx:range-negative-result}.

\bibliographystyle{alpha}
\bibliography{multicoding}

\appendix
\section{The RS compiler} \label{apx:rs}

%We recommend that the reader first read and digest the analysis in \cite{rs94} (specifically the proof of Lemma 5.1.1) before reading the proof of Lemma~\ref{lem:rs-failure-walk}.

\subsection{Description of compiler $\rs_0$}
The compiler described in \cite{rs94} is essentially the compiler that we used and referred to as the RS compiler, but there are a couple of technicalities that require us to modify the compiler. We begin by briefly describing a compiler $\rs_0$ for deterministic protocols. This compiler $\rs_0$ is a slight variant of the compiler described in \cite{rs94}, which we will modify still further in Section~\ref{sec:rs-technicalities} to define the final RS compiler. Fix an arbitrary digraph $G = (V, E)$ and some deterministic $T$-round protocol $\pi$ on $G$; we will describe $\rs_0(\pi)$. Our description of the simulation is not self-contained, and depends upon 
\cite{rs94} (specifically the proof of Lemma 5.1.1).
In terms of the description of $\Sigma$ in \cite{rs94}, the only change we are making is to set $k = \log|S|$, a constant, instead of having $k$ increase with the maximum indegree of $G$. This effectively eliminates the transmission code $\chi$.

By \cite[Lemma 1]{sch96}, there exists some ternary tree code\footnote{See \cite{sch96} for the definition of a tree code.} $\mathcal{T}$ of infinite depth, distance parameter $\frac{1}{2}$, and alphabet size $287$. This tree code will be used to encode strings over the alphabet $\{0, 1, \bkp\}$, and the tree code characters, which can be represented by bitstrings of length $9$, will be sent over the channels. We will refer collectively to the $9$ rounds needed to send a single tree code character as one \emph{step}.

In each step, a party $P_i$ begins by tree-decoding all the characters she's received so far from all of her in-neighbors, yielding, for each in-neighbor $P_j$, an \emph{estimated unparsed incoming transcript} $\hat{y}_{ji} \in \{0, 1, \bkp\}^*$. Each estimate $\hat{y}_{ji}$ is parsed into a \emph{estimated parsed incoming transcript} $\hat{w}_{ji} \in \{0, 1\}^*$ by processing from left to right, interpreting each $\bkp$ symbol as an instruction to delete the previous symbol. Similarly, for each out-neighbor $P_j$, $P_i$ recalls the string $y_{ij} \in \{0, 1, \bkp\}^*$ that she has transmitted to $P_j$, and parses this into a \emph{parsed outgoing transcript} $w_{ij} \in \{0, 1\}^*$. The \emph{parsed transcript} at $P_i$ at this moment is the collection of all these estimated parsed incoming transcripts and parsed outgoing transcripts.

We say that the parsed transcript at $P_i$ is \emph{consistent} if for every time $\tau$ and every out-neighbor $P_j$ of $P_i$, the $\tau$th bit in the parsed outgoing transcript $w_{ij}$ is the bit specified by $\pi$ to be sent to $P_j$, given the length-$(\tau - 1)$ prefixes of all the incoming transcripts. If the parsed transcript at $P_i$ is consistent, she transmits whatever bits are specified by $\pi$ (encoded using $\mathcal{T}$.) Otherwise, she transmits $\bkp$ to all of her out-neighbors. We run this simulation for $T_2$ steps, where $T_2$ is a parameter which will be chosen later.

\subsection{Analysis of $\rs_0(\pi)$}

We recall some terminology from \cite{rs94}. We say that an \emph{edge character error} occurs on the edge $(P_i, P_j)$ in step $\tau + 1$ if, in step $\tau$, $P_i$ sends some tree symbol, but $P_j$ receives a different tree symbol. We say that an \emph{edge tree error} occurs on $(P_i, P_j)$ in step $\tau + 1$ if, in step $\tau + 1$, $P_j$'s estimated, unparsed transcript $\hat{y}_{ij}$ differs from the true unparsed transcript $y_{ij}$. We simply say that a \emph{tree error} occurs at $P_j$ in step $\tau + 1$ if, for some in-neighbor $P_i$, an edge tree error occurs on $(P_i, P_j)$ in step $\tau + 1$. We say that $(P_i, \tau)$ and $(P_j, \tau')$ are \emph{time-like} if there is a path from $P_i$ to $P_j$ of length no more than $\tau' - \tau$. For edges $e, e'$, we say that $(e, \tau)$ and $(e', \tau')$ are \emph{time-like} if there is some walk which begins with $e$, ands with $e'$, and has length no more than $\tau' - \tau + 1$. For example, $(e, \tau)$ and $(e, \tau')$ are time-like for any $\tau' \geq \tau$; if $e = (P_i, P_j)$ and $e' = (P_j, P_k)$, then $(e, \tau)$ and $(e', \tau + 1)$ are time-like. Finally, we say that $(e, \tau)$ and $(P_i, \tau')$ are time-like if there is a walk which begins with $e$, ends at $P_i$, and has length no more than $\tau' - \tau + 1$. For example, if $e = (P_i, P_j)$, then $(e, \tau)$ and $(P_j, \tau)$ are time-like. A \emph{time-like sequence} is a sequence where each pair of successive elements is a time-like pair. The \emph{time history cone} of $(P_i, \tau)$ is the set of all $(P_j, \tau')$ such that $(P_j, \tau')$ and $(P_i, \tau)$ are time-like, unioned with the set of all $(e, \tau')$ such that $(e, \tau')$ and $(P_i, \tau)$ are time-like.

For a party $P_i$ and a time $\tau$, $RP(P_i, \tau)$ is the number of rounds $t$ such that the first $t$ rounds in all of $P_i$'s parsed outgoing transcripts (at the end of step $\tau$) match what she would send if all the parties followed $\pi$ on a noiseless network. As in \cite{rs94}, our goal is to show that if the number of errors is sufficiently small, then $RP(P_i, \tau)$ will be close to $\tau$ for every $P_i$ and every $\tau$.\footnote{Observe that this will not immediately mean that the simulation is successful, since the criterion for success is that the parties give the correct outputs, which depends on their estimated incoming transcripts, not their outgoing transcripts. We deal with this small technicality in Section~\ref{sec:rs-technicalities}.}

\paragraph{Intuition.} A key part of the intuition in \cite{rs94} is that if two errors have a space-like separation, then they should cause no more delay than if only one of them occurred. The analysis in \cite{rs94} establishes a quantitative version of this idea with regard to tree errors. Specifically, it is shown \cite[Lemma 5.1.1]{rs94} that if $RP(P_i, \tau) = \tau - \ell$, then there is a time-like sequence of at least $\ell/2$ tree errors in the time history cone of $P_i$ at $\tau$. But the intuition still applies if we look at the underlying edge tree errors, instead of just tree errors.

For example, suppose $G$ is a star graph, with all edges but one directed inward; say there are $q$ edges pointing inward. Consider the adversarial strategy of dividing the rounds of the simulation into $q$ equal segments, and spending the $i$th segment zeroing out the bits sent across the $i$th inward-facing edge. As tree errors (rather than edge tree errors), this is a time-like sequence of errors, simply because they all have the same recipient. Thus, the analysis in \cite{rs94} does not show that the simulation will succeed in the face of this adversary. However, if we pay attention to the underlying edge tree errors, we see that these errors do not have a time-like separation; signals sent along one incoming edge can never affect signals sent along another incoming edge. This suggests that for sufficiently large $q$, the simulation will succeed (for all round complexities $T$), since the longest sequence of time-like edge tree errors is only of length about $T/q$.

And indeed, this suggestion is easily seen to be true. After the central party recovers from the edge tree errors on the first couple of incoming edges, she is sufficiently far behind the other incoming parties in the simulation that further edge errors do not affect her; by the time a symbol actually affects the central party's transmissions, it has been ``cleaned up'' by the tree code mechanism, so she makes no further mistakes.

For a party $P_i$ and a time $\tau_0$, we define $Y(P_i, \tau_0)$ to be the maximum length of any time-like sequence of edge character errors in the time history cone of $P_i$ at $\tau_0$. (Compare $Y(P_i, \tau)$ to the quantity $X(P_i, \tau)$ analyzed in \cite{rs94}.) As in \cite{rs94}, we let $B(P_i, \tau)$ denote the number of times that $P_i$ has transmitted $\bkp$ up to step $\tau$, and we let $AT(P_i, \tau) = \tau - 2B(P_i, \tau)$, so that $AT(P_i, \tau)$ is the length of every outgoing transcript of $P_i$'s at time $\tau$. Most of our effort will go toward proving the following proposition, analogous to \cite[Proposition 5.2.1]{rs94}.

\begin{prop} \label{prop:rs-careful-analysis}
For any party $P_i$ and any time $\tau$,
\begin{equation} \label{eqn:rs-careful-analysis}
\tau \leq RP(P_i, \tau) + 24Y(P_i, \tau) + B(P_i, \tau).
\end{equation}
\end{prop}
Toward proving Proposition~\ref{prop:rs-careful-analysis}, we make the following definitions.
\begin{defn}
For an edge $(P_i, P_j) \in E$ and an alleged (parsed) transcript $z_{ij} \in \{0, 1\}^*$, we define the \emph{accuracy} of $z_{ij}$ to be the length of the longest prefix of $z_{ij}$ which is a prefix of the sequence of bits that $P_i$ would send $P_j$ on a noiseless network.
\end{defn}
\begin{defn}
The \emph{action} of $P_j$ in step $\tau + 1$ is defined as follows.\footnote{This notion is analogous to, but not the same as, the function {\sc Action} in \cite{rs94}.}
\begin{itemize}
\item Suppose $RP(P_j, \tau + 1) > RP(P_j, \tau)$. Then the action of $P_j$ in step $\tau + 1$ is \texttt{progress}.
\item Suppose $RP(P_j, \tau + 1) = RP(P_j, \tau)$, and $P_j$ backs up in step $\tau + 1$. Then the action of $P_j$ in step $\tau + 1$ is \texttt{justified backup}.
\item Suppose $RP(P_j, \tau + 1) < RP(P_j, \tau)$, or $RP(P_j, \tau + 1) = RP(P_j, \tau)$ and $P_j$ transmits data in step $\tau + 1$. Say that an in-neighbor $P_i$ of $P_j$ is \emph{accuracy minimizing} if the accuracy of the estimated, parsed transcript $\hat{w}_{ij}$ is minimized at $P_i$ among in-neighbors of $P_j$.
\begin{itemize}
\item If, for every accuracy minimizing $P_i$, the accuracy of the true transcript $w_{ij}$ is strictly more than the accuracy of the estimated transcript $\hat{w}_{ij}$, we say that the action of $P_j$ in step $\tau + 1$ is \texttt{harmful tree error}.
\item Otherwise, we say that the action of $P_j$ in step $\tau + 1$ is \texttt{harmful propagated error}.
\end{itemize}
\end{itemize}
\end{defn}

Lemmas~\ref{lem:rp-k},~\ref{lem:tree-error-cases}, and~\ref{lem:backup-implies-behind} are fairly technical, and are motivated only by the fact that they will be useful for proving Proposition~\ref{prop:rs-careful-analysis}.
\begin{lem} \label{lem:rp-k}
Suppose that in step $\tau_0 + 1$, the action of $P_j$ is either \emph{\texttt{harmful tree error}} or \emph{\texttt{harmful propagated error}}. Let $P_i$ be accuracy minimizing, and say that the accuracy of $\hat{w}_{ij}$ is $k - 1$. Then if $P_j$ transmitted data in step $\tau_0 + 1$, then $RP(P_j, \tau_0) \geq k$, while if $P_j$ backed up in step $\tau_0 + 1$, then $RP(P_j, \tau_0) \geq k + 1$.
\end{lem}

\begin{proof}
First, suppose $P_j$ transmitted data in step $\tau_0 + 1$. Then $P_j$ did not ``regret'' any of her transmissions, i.e. her parsed transcript was consistent. Thus, the length-$k$ prefixes of her outgoing parsed transcripts must be correct, since they match the accurate length-$(k - 1)$ prefixes of her estimated incoming parsed transcripts. Next, suppose $P_j$ backed up in step $\tau_0 + 1$. From the action, we must have $RP(P_j, \tau_0) = AT(P_j, \tau_0)$. Furthermore, since $P_j$ ``regrets'' a transmission which is, in fact, correct, we must have $k \leq AT(P_j, \tau_0) - 1$. Therefore, $RP(P_j, \tau_0) \geq k + 1$ as claimed.
\end{proof}

Recall that the \emph{magnitude} of an edge tree error on $(P_i, P_j)$ is the length of the suffix of the affected estimated unparsed transcript $\hat{y}_{ij}$ which begins with the first symbol which differs from the corresponding symbol of the true unparsed transcript $y_{ij}$.

\begin{lem} \label{lem:tree-error-cases}
Suppose Equation~$\ref{eqn:rs-careful-analysis}$ holds for every party $P_i$ and every $\tau \leq \tau_0$. Suppose that in step $\tau_0 + 1$, there is an edge tree error of magnitude $M > 0$ on $(P_i, P_j)$. Suppose that $RP(P_i, \tau_0) \leq RP(P_j, \tau_0) + 2M$, and $B(P_i, \tau_0) \leq B(P_j, \tau_0) + M$. Then Equation~$\ref{eqn:rs-careful-analysis}$ holds for $P_j$ and $\tau = \tau_0 + 1$.
\end{lem}

\begin{proof}
Since $RP$ and $B$ change at most one each round, we have
\begin{align}
RP(P_i, \tau_0 + 1 - M) &\leq RP(P_i, \tau_0) + M - 1 \\
&\leq RP(P_j, \tau_0) + 3M - 1 \\
&\leq RP(P_j, \tau_0 + 1) + 3M, \label{eqn:tree-error-rp}
\end{align}
and similarly
\begin{equation}\label{eqn:tree-error-b}
B(P_i, \tau_0 + 1 - M) \leq B(P_j, \tau_0 + 1) + 2M.
\end{equation}
From the tree code condition, we can be sure that in the $M$ steps preceding and including step $\tau_0 + 1$, the character error rate on $(P_i, P_j)$ was at least $\frac{1}{4}$. Therefore, we have
\begin{equation}\label{eqn:tree-error-y}
Y(P_i, \tau_0 + 1 - M) \leq Y(P_j, \tau_0 + 1) - \frac{1}{4} M.
\end{equation}
Applying Equation~\ref{eqn:rs-careful-analysis} to $P_i$ at $\tau = \tau_0 + 1 - M$ and using Equations~\ref{eqn:tree-error-rp}, \ref{eqn:tree-error-b}, and \ref{eqn:tree-error-y}, we have
\begin{align}
\tau_0 + 1 - M &\leq RP(P_i, \tau_0 + 1 - M) + 24Y(P_i, \tau_0 + 1 - M) + B(P_i, \tau_0 + 1 - M) \\
&\leq RP(P_j, \tau_0 + 1) + 3M + 24Y(P_j, \tau_0 + 1) - 6M + B(P_j, \tau_0 + 1) + 2M \\
&= RP(P_j, \tau_0 + 1) + 24Y(P_j, \tau_0 + 1) + B(P_j, \tau_0 + 1) - M,
\end{align}
which completes the proof.
\end{proof}

The following lemma is proven in exactly the same way as an analogous statement in \cite{rs94}.
\begin{lem} \label{lem:backup-implies-behind}
Suppose Equation~$\ref{eqn:rs-careful-analysis}$ holds for all $P_i$ and all $\tau \leq \tau_0$. Suppose that $RP(P_j, \tau_0) = AT(P_j, \tau_0)$, and for some in-neighbor $P_i$, $B(P_i, \tau_0) > B(P_j, \tau_0)$. Then Equation~$\ref{eqn:rs-careful-analysis}$ holds for $P_j$ and $\tau = \tau_0 + 1$.
\end{lem}

\begin{proof}
By hypothesis, $RP(P_j, \tau_0) = \tau_0 - 2B(P_j, \tau_0)$. We also have $RP(P_i, \tau_0) \leq \tau_0 - 2B(P_i, \tau_0)$. Subtracting gives $RP(P_j, \tau_0) + 2B(P_j, \tau_0) \geq RP(P_i, \tau_0) + 2B(P_i, \tau_0)$. Using our assumption $B(P_i, \tau_0) \geq B(P_j, \tau_0) + 1$, this implies
\begin{equation}
RP(P_j, \tau_0) + B(P_j, \tau_0) \geq RP(P_i, \tau_0) + B(P_i, \tau_0) + 1.
\end{equation}
Now, observe that $RP(P_j, \tau_0 + 1) + B(P_j, \tau_0 + 1) \geq RP(P_j, \tau_0) + B(P_j, \tau_0)$, because the $B$ term can only increase in step $\tau_0 + 1$, while the $RP$ term can only decrease by at most $1$ and in that case the $B$ term increased. Therefore, we have
\begin{equation}
RP(P_i, \tau_0) + B(P_i, \tau_0) < RP(P_j, \tau_0 + 1) + B(P_j, \tau_0 + 1).
\end{equation}
Of course, $Y(P_i, \tau_0) \leq Y(P_j, \tau_0 + 1)$, so an application of Equation~\ref{eqn:rs-careful-analysis} at $P_i$ at $\tau = \tau_0$ completes the proof.
\end{proof}

\begin{proof}[of Proposition~$\ref{prop:rs-careful-analysis}$]
We proceed by induction on $\tau$. At $\tau = 0$, all terms of Equation~\ref{eqn:rs-careful-analysis} are zero. For the inductive step, assume that Equation~\ref{eqn:rs-careful-analysis} holds for all $\tau < \tau_0$; we will prove that it holds for $\tau = \tau_0 + 1$. Fix some party $P_i$. If the action of $P_i$ in step $\tau + 1$ is \texttt{progress}, then $RP(P_i, \tau_0 + 1) = RP(P_i, \tau_0) + 1$, and the $B$ and $Y$ terms do not decrease from $\tau_0$ to $\tau_0 + 1$, so we are done. Similarly, if the action of $P_i$ is \texttt{justified backup}, then $B(P_i, \tau_0 + 1) = B(P_i, \tau_0) + 1$, and the $RP$ and $Y$ terms do not decrease. The final two cases, \texttt{harmful tree error} and \texttt{harmful propagated error}, are treated in Lemmas~\ref{lem:hte} and~\ref{lem:hpe} below.%\phantom\qedhere

\begin{lem} \label{lem:hte}
Suppose Equation~$\ref{eqn:rs-careful-analysis}$ holds for all $P_i$ and all $\tau \leq \tau_0$. Suppose that in step $\tau_0 + 1$, the action of $P_j$ is \emph{\texttt{harmful tree error}}. Then Equation~$\ref{eqn:rs-careful-analysis}$ holds for $P_j$ and $\tau = \tau_0 + 1$.
\end{lem}

\begin{proof}
Let $P_i$ be an accuracy minimizing in-neighbor of $P_j$, and say that the accuracy of $\hat{w}_{ij}$ is $k - 1$. As the name of the action indicates, because the $k$th symbol of the true transcript $w_{ij}$ is correct and present while the $k$th symbol of the estimated transcript $\hat{w}_{ij}$ is not, there must have been an edge tree error on $(P_i, P_j)$. Say that the tree error was of magnitude $M$. Then $w_{ij}$ and $\hat{w}_{ij}$ agree in their first $(|w_{ij}| - 2M)$ positions. In particular, $|w_{ij}| \leq k + 2M - 1 < k + 2M$.

By Lemma~\ref{lem:rp-k}, $RP(P_j, \tau_0) \geq k$. Therefore, $|w_{ij}| \leq RP(P_j, \tau_0) + 2M$. Of course, $|w_{ij}| = AT(P_i, \tau_0) \geq RP(P_i, \tau_0)$, so
\begin{equation}
RP(P_i, \tau_0) \leq RP(P_j, \tau_0) + 2M.
\end{equation}
Again because $w_{ij}$ and $\hat{w}_{ij}$ agree in their first $(|w_{ij}| - 2M)$ positions, $|w_{ij}| \geq |\hat{w}_{ij}| - 2M$.

First, suppose that $P_j$ transmitted data in step $\tau_0 + 1$. Then $|\hat{w}_{ij}| \geq AT(P_j, \tau_0)$, so $AT(P_i, \tau_0) \geq AT(P_j, \tau_0) - 2M$, which implies that
\begin{equation}
B(P_i, \tau_0) \leq B(P_j, \tau_0) + M.
\end{equation}
An application of Lemma~\ref{lem:tree-error-cases} completes the proof in this case. Next, suppose that $P_j$ backed up in step $\tau_0 + 1$. If $B(P_i, \tau_0) \leq B(P_j, \tau_0)$, then once again, Lemma~\ref{lem:tree-error-cases} completes the proof. But if $B(P_i, \tau_0) > B(P_j, \tau_0)$, then Lemma~\ref{lem:backup-implies-behind} completes the proof, since the action implies that $AT(P_j, \tau_0) = RP(P_j, \tau_0)$.
\end{proof}

\begin{lem} \label{lem:hpe}
Suppose Equation~$\ref{eqn:rs-careful-analysis}$ holds for every party $P_i$ and every $\tau \leq \tau_0$. Suppose that in step $\tau_0 + 1$, the action of $P_j$ is \emph{\texttt{harmful propagated error}}. Then Equation~$\ref{eqn:rs-careful-analysis}$ holds for $P_j$ and $\tau = \tau_0 + 1$.
\end{lem}

\begin{proof}
Let $P_i$ be an accuracy minimizing in-neighbor such that the accuracy of $w_{ij}$ is no more than the accuracy of $\hat{w}_{ij}$, and say that the accuracy of $\hat{w}_{ij}$ is $k - 1$. Then $RP(P_i, \tau_0) < k$. We claim that
\begin{equation} \label{eqn:rp-behind}
RP(P_i, \tau_0) \leq RP(P_j, \tau_0 + 1) - 1.
\end{equation}
To see why, first suppose $P_j$ transmitted data in step $\tau_0 + 1$; then $RP(P_j, \tau_0 + 1) = RP(P_j, \tau_0)$, and by Lemma~\ref{lem:rp-k}, $RP(P_j, \tau_0) \geq k$, which completes the proof of Equation~\ref{eqn:rp-behind}. Next, suppose $P_j$ backed up in step $\tau_0 + 1$; then $RP(P_j, \tau_0 + 1) = RP(P_j, \tau_0) - 1$, and by Lemma~\ref{lem:rp-k}, $RP(P_j, \tau_0) \geq k + 1$, so $RP(P_j, \tau_0 + 1) \geq k$, again completing the proof of Equation~\ref{eqn:rp-behind}.

Now, for a first case, suppose $B(P_i, \tau_0) \leq B(P_j, \tau_0 + 1)$. Then applying Equation~\ref{eqn:rs-careful-analysis} to $P_i$ at $\tau = \tau_0$ gives
\begin{align}
\tau_0 &\leq RP(P_i, \tau_0) + 24Y(P_i, \tau_0) + B(P_i, \tau_0) \\
&\leq RP(P_j, \tau_0 + 1) - 1 + 24Y(P_j, \tau_0 + 1) + B(P_j, \tau_0 + 1),
\end{align}
completing the proof. For the second case, suppose $B(P_i, \tau_0) = B(P_j, \tau_0 + 1) + \ell$, with $\ell > 0$.
\begin{itemize}
\item For the first subcase, suppose that $P_j$ transmitted data in step $\tau_0 + 1$; then $B(P_i, \tau_0) = B(P_j, \tau_0) + \ell$, so there was an edge tree error of magnitude $M \geq \ell$ on edge $(P_i, P_j)$ in step $\tau_0 + 1$. Therefore, we can apply Lemma~\ref{lem:tree-error-cases} to complete the proof, since $RP(P_i, \tau_0) \leq RP(P_j, \tau_0 + 1) - 1$.
\item Finally, for the second subcase, suppose that $P_j$ backed up in step $\tau_0 + 1$. Then $RP(P_j, \tau_0) = AT(P_j, \tau_0)$ and $B(P_i, \tau_0) > B(P_j, \tau_0 + 1) >  B(P_j, \tau_0)$, so we can apply Lemma~\ref{lem:backup-implies-behind}.
\end{itemize}
This completes the proofs of Lemma~\ref{lem:hpe} and Proposition~\ref{prop:rs-careful-analysis}.
\end{proof}
\end{proof}

The following lemma follows easily from Proposition~\ref{prop:rs-careful-analysis}; it is analogous to \cite[Lemma 5.1.1]{rs94}.
\begin{lem} \label{lem:delay-walk}
Suppose $RP(P_i, \tau) = \tau - \ell$. Then there is some time-like sequence of $\frac{\ell}{48}$ edge character errors in the time history cone of $P_i$ at $\tau$.
\end{lem}

\begin{proof}
Doubling Equation~\ref{eqn:rs-careful-analysis} and rearranging gives
\[
RP(P_i, \tau) - \tau + 48Y(P_i, \tau) \geq \tau - RP(P_i, \tau) - 2B(P_i, \tau) = AT(P_i, \tau) - RP(P_i, \tau).
\]
The right-hand side is nonnegative, so $Y(P_i, \tau) \geq \frac{\ell}{48}$.
\end{proof}

\subsection{The final RS compiler} \label{sec:rs-technicalities}
All the work we have done to analyze $\rs_0$ has focused on $RP$ as a measure of progress, but $RP$ is not directly related to our success criterion, which is that the parties give the correct outputs at the end of the simulation, or equivalently that the parties are able to correctly guess the bits they would have \emph{received} if the parties had followed the original protocol on a noiseless network. To address this technicality, before applying $\rs_0$, we will modify the protocol so that any bits that the parties receive are immediately retransmitted over dummy channels to dummy parties. The resulting compiler is the RS compiler, which we denote $\rs$.

We describe $\rs(\pi)$ for arbitrary deterministic protocols $\pi$. Let $\bar{V} = V \cup \{P_{n + i} : 1 \leq i \leq n\}$ be two disjoint copies of $V$, where $P_{n + i}$ is a copy of $P_i$. Let $\bar{E} = E \cup \{(P_i, P_{n + j}) : (P_j, P_i) \in E\}$. Let $\bar{G} = (\bar{V}, \bar{E})$. From the $T$-round protocol $\pi$ on $G$, we define a $(T + 1)$-round protocol $\bar{\pi}$ on $\bar{G}$ as follows. The inputs for $\bar{\pi}$ are exactly the same as the inputs for $\pi$. When $P_i \in V$ receives an input $x_i$, she uses the ``ordinary'' edges in $G$ to do what $\pi$ instructs her. (In the extra round at the end, she just sends a zero on every ordinary edge.) On each ``dummy'' edge $(P_i, P_{n + j})$, in round $\tau$, she transmits the bit that she received in round $\tau - 1$ on the corresponding ordinary edge $(P_j, P_i)$. (In the first round, she just sends a zero on every dummy edge.)

The protocol $\rs(\pi)$, which runs on the original graph $G$, is essentially $\rs_0(\bar{\pi})$, with $T_2 = 2T + 1$. Technically, $\rs_0(\bar{\pi})$ runs on $\bar{G}$; naturally, the parties in $\rs(\pi)$ do not literally send bits across the dummy edges, but they keep track of which bits they would have sent across the dummy edges, and on the ordinary edges, they behave exactly as in $\rs_0(\bar{\pi})$. At the end of the simulation, $P_i$ gives as output whatever $\pi$ instructs her to give as output, under the assumption that her parsed outgoing transcripts accurately reflect the (incoming and outgoing) transcripts that would have occurred if the parties had followed $\pi$ on a noiseless network.

We can now prove Lemma~\ref{lem:rs-failure-walk} and Proposition~\ref{prop:rs94}.

\begin{proof}[of Lemma~$\ref{lem:rs-failure-walk}$]
The round complexity of $\rs(\pi^*[G, T])$ is $9T_2 \leq 27T$. If the execution fails, then there must be some party $P_i$ such that $RP(P_i, 2T + 1) < T + 1$. Therefore, by Lemma~\ref{lem:delay-walk}, there was some time-like sequence of $\frac{T}{48}$ edge character errors. Each character error is associated with at least one bit error.
\end{proof}

\begin{proof}[of Proposition~$\ref{prop:rs94}$]
By Lemma~\ref{lem:rs-failure-walk}, if $\rs(\pi^*)$ fails, there were at least $\frac{T}{48}$ bit errors, which is a bit error rate of at least $\rho = \frac{T/48}{27Tm} = \frac{1}{1296m}$.
\end{proof}

\section{Proof of Lemma~\ref{lem:flow-to-paths}} \label{apx:flow-to-paths}

Without loss of generality, assume that precisely $\lambda$ units of each commodity flow in $F$. For each $(P_i, P_j)$, randomly form a path $p_{ij}$ from $P_i$ to $P_j$ as follows. Initially, $p_{ij}$ is just $P_i$; in each step, extend $p_{ij}$ by randomly selecting an outgoing edge, with the probability of selecting an edge being proportional to the amount of flow of commodity $(P_i, P_j)$ which goes across that edge. Repeat until the path reaches $P_j$. Let $\mathcal{P}$ be the set of paths $p_{ij}$ formed in this way.

Fix some edge $e \in E$. For each $(P_i, P_j) \in E$, the probability that $e \in p_{ij}$ is equal to the flow of commodity $(P_i, P_j)$ across $e$ in $F$ divided by $\lambda$, and these are independent events. Therefore, the expected value of the congestion $c_e$ of $e$ is exactly equal to $f_e / \lambda$, where $f_e$ is the total flow across $e$ in $F$, and by the Chernoff bound, for any $\epsilon > 0$,
\begin{equation}
\Pr\left(c_e \geq (1 + \epsilon) \frac{f_e}{\lambda}\right) \leq \exp\left(-\frac{\epsilon^2}{2 + \epsilon} \frac{f_e}{\lambda} \right) \leq \exp\left(-\frac{\epsilon^2}{(2 + \epsilon)^2} \cdot (1 + \epsilon) \frac{f_e}{\lambda}\right).
\end{equation}
If we define $\epsilon$ so that $(1 + \epsilon) \frac{f_e}{\lambda} = 9(\frac{1}{\lambda} + \ln m)$, then certainly $\epsilon > 1$ simply because $f_e \leq 1$, and hence $\frac{\epsilon^2}{(2 + \epsilon)^2} > \frac{1}{9}$. Therefore,
\begin{equation}
\Pr\left(c_e \geq 9\left(\frac{1}{\lambda} + \ln m\right)\right) < \exp\left(-\frac{1}{9} \cdot 9\left(\frac{1}{\lambda} + \ln m\right)\right) = \frac{1}{m} e^{-1/\lambda}.
\end{equation}
Taking a union bound over all $m$ edges $e$, we see that with positive probability, the total congestion of $\mathcal{P}$ is no more than $9(\frac{1}{\lambda} + \ln m)$. \qed

\section{Proof of Theorem~\ref{thm:magi}} \label{apx:magi}

The compiler $C$ used to prove Theorem~\ref{thm:magi} is formed by composing the RS compiler with a \emph{magi coding}\footnote{To avoid King Herod, the biblical Magi went home along a different route than they had planned. (Matthew 2:12)} compiler, as in Equation~\ref{eqn:magi}. Note that the RS compiler comes first in this composition, in contrast to the compilers used to prove Theorems~\ref{thm:sparsifier-error-correction} and~\ref{thm:reachability-preserving-subgraph-error-rate}.
\begin{equation} \label{eqn:magi}
C: \pi^* \quad \stackrel{\text{RS compiler}}{\longmapsto} \quad \pi' \quad \stackrel{\text{Magi coding compiler}}{\longmapsto} \quad \tilde{\pi}^*
\end{equation}

We will need the following fact about the RS compiler, which is stronger than Proposition~\ref{prop:rs94}.
\begin{prop} \label{prop:rs-rer}
There exists $\eta > 0$ such that if an execution of $\rs(\pi^*)$ fails, then in that execution, the fraction of rounds in which bit errors occurred was at least $\eta$.
\end{prop}

\begin{proof}
This is established already by the proof of Lemma~\ref{lem:rs-failure-walk}, since time-like character errors must occur in distinct steps.
\end{proof}

The idea of the magi coding compiler is straightforward: parties send bits to randomly chosen third parties, who deliver them to their recipients.

\subsection{Description of magi coding}

Since $\pi'$ is deterministic, immediately upon receiving her input, every party $P_i$ can compute the transmission function $x_i'$ which describes her behavior in $\pi'$. Magi coding works by simulating each round of $\pi'$ individually; the simulation of a single round is given by Algorithm~\ref{alg:magi}. That is, at the beginning of an execution of Algorithm~\ref{alg:magi}, for each edge $(P_i, P_j)$, $P_i$ has in mind a bit $b_{ij}$ that she would like to send to $P_j$; at the end of the execution, $P_j$ has a guess $\hat{b}_{ij}$ about the value of $b_{ij}$. If we have inductively simulated $\tau$ rounds of $\pi'$ in this way, then in round $\tau + 1$, each bit $b_{ij}$ is determined by $x_i'$ under the assumption that the previous estimates $\{\hat{b}_{j'i}\}$ were all correct.

Say that every adjacent pair of vertices have at least $\epsilon n$ common neighbors. Algorithm~\ref{alg:magi} makes reference to a number $\ell$. We define
\begin{equation}
\ell = \frac{24}{\epsilon} \log\left(\frac{2m}{\alpha}\right),
\end{equation}
where $\alpha$ is a parameter to be chosen later. For our purposes, it will suffice to take $\alpha = \frac{1}{4} \eta$, where $\eta$ is that given in Proposition~\ref{prop:rs-rer}.

\begin{algorithm} \label{alg:magi}
\Loop{$\ell$ times}
{
Using shared randomness, pick a random number $1 \leq r \leq n$, as well as two random bits $w_{ij}, w_{ij}'$ for each edge $(P_i, P_j) \in E$. \;
\For{each edge $(P_i, P_k) \in E$, all simultaneously}
{
  Define $j$ so that $i + j + k \equiv r \pmod{n}$. \;
  $P_i$ sends bit $b_{ij} \oplus w_{ij}$ to $P_k$, who receives $y_{ij} = b_{ij} \oplus w_{ij} \oplus \texttt{noise}$. \;
}
\For{each edge $(P_k, P_j) \in E$, all simultaneously}
{
  Define $i$ so that $i + j + k \equiv r \pmod{n}$. \;
  $P_k$ sends bit $y_{ij} \oplus w_{ij}'$ to $P_k$, who receives $z_{ij} = b_{ij} \oplus w_{ij} \oplus \texttt{noise} \oplus w_{ij}' \oplus \texttt{more noise}$. \;
  $P_j$ casts a vote for $z_{ij} \oplus w_{ij} \oplus w_{ij}'$, in an election for the office of $\hat{b}_{ij}$ with candidates $\{0, 1\}$.
}
}
By majority vote, for each $(P_i, P_j) \in E$, $P_j$ decides on an estimate $\hat{b}_{ij}$.
\caption{A single segment of magi coding.}
\end{algorithm}

We refer to the $2\ell$ rounds of communication needed to execute Algorithm~\ref{alg:magi} as one \emph{segment}. The simulation runs for $T'$ segments, where $T'$ is the round complexity of $\pi'$.

\subsection{Analysis of magi coding}

We say that a \emph{simulated bit error} occurs in segment $t$ if, at the end of segment $t$, for some edge $(P_i, P_i)$, $b_{ij} \neq \hat{b}_{ij}$.

\begin{lem} \label{lem:simulated-bit-error-difficulty}
Fix a segment. Suppose that in that segment, the adversary introduces at most $\frac{1}{32}\epsilon n \ell$ bit errors. Then the probability of a simulated bit error in that segment is no more than $\alpha$.
\end{lem}

\begin{proof}
Fix $(P_i, P_j) \in E$. In each iteration of the loop, the probability that $r$ is chosen such that $i + j + r \pmod{n}$ is the index of a common neighbor of $P_i$ and $P_j$ is at least $\epsilon$. Thus, the expected number of votes that $P_j$ casts regarding $\hat{b}_{ij}$ is $\epsilon \ell$. These are independent events, so by the Chernoff bound, the probability that fewer than $\frac{1}{2} \epsilon \ell$ such votes are cast is no more than $\exp(-\frac{1}{8} \epsilon \ell) \leq \frac{\alpha}{2m}$.

Say the number of bit errors that the adversary introduces during the $\tau$th iteration of the main loop of Algorithm~\ref{alg:magi} is $a_\tau$, for $1 \leq \tau \leq \ell$. Fix some iteration $\tau$ of that loop. For any edge $e$, the probability that $b_{ij}$ is sent across $e$ during that iteration is no more than $\frac{2}{n}$, because of the choice of $r$. Furthermore, in both rounds in that iteration, it remains true conditioned on all bits that have been transmitted (i.e. on all the information that the adversary has) that the probability that $b_{ij}$ is sent across $e$ in that round is no more than $\frac{2}{n}$. (This was the purpose of the random bits $w_{ij}$, $w'_{ij}$.) Therefore, by the union bound, the probability that the adversary corrupts $b_{ij}$ in this iteration is at most $\frac{4 a_\tau}{n}$.

The expected number of iterations in which the adversary corrupts $b_{ij}$ is no more than $\frac{1}{8} \epsilon \ell$. These events are independent, so by the Chernoff bound, the probability that more than $\frac{1}{4} \epsilon \ell$ incorrect votes are cast is no more than $\exp(-\frac{1}{24} \epsilon \ell) = \frac{\alpha}{2m}$.

If, at the end of the segment, $\hat{b}_{ij} \neq b_{ij}$, then either $P_j$ cast fewer than $\frac{1}{2} \epsilon \ell$ votes regarding $\hat{b}_{ij}$, or else $P_j$ cast at least $\frac{1}{4} \epsilon \ell$ incorrect votes regarding $\hat{b}_{ij}$. Therefore, by the union bound, $\Pr(\hat{b}_{ij} \neq b_{ij}) \leq \frac{\alpha}{m}$. Taking another union bound over the $m$ edges $(P_i, P_j)$ completes the proof.
\end{proof}

Observe that magi coding does not simply make it difficult for the adversary to create a high simulated bit error rate. Indeed, if she wants to introduce a simulated bit error rate of $\rho$, she can simply choose a $\rho$ fraction of the segments and corrupt every single edge in the even-numbered rounds of the chosen segments, costing her an actual bit error rate of $\rho / 2$. Rather, the gain from magi coding is that it makes it difficult for the adversary to introduce a positive number of simulated bit errors in a large number of segments:

\begin{lem} \label{lem:tiny-failure-probability}
Suppose that during the execution of $\tilde{\pi}^*$, the adversary introduces at most $\frac{1}{16} \alpha \epsilon \ell n T'$ bit errors. Then the probability that at least one simulated bit error occurs in each of $4 \alpha T'$ different segments is $e^{-\Omega((1 - 2\alpha)T')}$.
\end{lem}

\begin{proof}
Say that a segment is \emph{targeted} if the adversary introduces at least $\frac{1}{32} \epsilon \ell n$ bit errors in that segment. By hypothesis, at most $2 \alpha T'$ segments are targeted. On the other hand, from Lemma~\ref{lem:simulated-bit-error-difficulty}, we know that in each non-targeted segment, the probability that at least one simulated bit error occurs is at most $\alpha$. There are $T'$ segments total, so the expected number of non-targeted segments in which at least one simulated bit error occurs is at most $\alpha T'$. There are at least $(1 - 2\alpha)T'$ non-targeted segments, and these events are all independent. Therefore, by the Chernoff bound, the probability that at least one simulated bit error occurs in $2 \alpha T'$ different non-targeted segments is $e^{-\Omega((1 - 2\alpha)T')}$.
\end{proof}

\begin{proof}[of Theorem~$\ref{thm:magi}$]
Note that $\tilde{\pi}^*$ runs in $2\ell T'$ rounds, which is $\O(T \log n)$ rounds as claimed. If $\tilde{\pi}^*$ fails, then by Proposition~\ref{prop:rs-rer}, there were at least $\eta T'$ distinct segments in which a simulated bit error occurred. By Lemma~\ref{lem:tiny-failure-probability} with $\alpha = \frac{1}{4} \eta$, as long as the adversary is restricted to introducing at most $\frac{1}{16} \alpha \epsilon \ell n T'$ bit errors, then the probability of simulated bit errors occurring in $\eta T'$ distinct segments is $e^{-\Omega(T)}$. This tolerable amount of error corresponds to the bit error rate $\rho$ given by
\begin{equation}
\rho = \frac{\frac{1}{16} \frac{1}{4} \eta \epsilon \ell n T'}{m \tilde{T}} = \frac{\eta \epsilon n}{128 m} \geq \frac{\eta \epsilon}{128 n}.
\end{equation}
\end{proof}

\section{Illustration of the proof of Theorem~\ref{thm:signal-diameter-negative-result}} \label{apx:signal-diameter-figs}
\begin{figure}[H]
\centering

\begin{tikzpicture}
\node[truecircle] (V0) at (0, 2) {$V_0$};
\node[trueellipse] (V1) at (3, 2) {$V_1$};
\node[trueellipse] (V2) at (6, 2) {$V_2$};

\node[falsecircle] (V0b) at (0, -2) {$\bar{V}_0$};
\node[falseellipse] (V1b) at (3, -2) {$\bar{V}_1$};
\node[falseellipse] (V2b) at (6, -2) {$\bar{V}_2$};

\node[confusedellipse] (V3) at (9, 0) {$V_3$};
\node[confusedellipse] (V4) at (12, 0) {$V_4$};
\node[confusedellipse] (V5) at (15, 0) {$V_5$};

\coordinate (splitpoint1) at ($(V3.north west)+(-0.3, 0)$) {};
\coordinate (splitpoint2) at ($(V3.west)+(-0.3, 0)$) {};
\coordinate (splitpoint3) at ($(V3.south west)+(-0.3, 0)$) {};

\draw[trueedge] (V0.north east) to (V1.north west);
\draw[trueedge] (V1.west) to (V0.east);
\draw[trueedge] (V0.south east) to (V1.south west);
\draw[trueedge] (V1.north east) to (V2.north west);
\draw[trueedge] (V2.south west) to (V1.south east);

\draw[falseedge] (V0b.north east) to (V1b.north west);
\draw[falseedge] (V1b.west) to (V0b.east);
\draw[falseedge] (V0b.south east) to (V1b.south west);
\draw[falseedge] (V1b.north east) to (V2b.north west);
\draw[falseedge] (V2b.south west) to (V1b.south east);

\draw[truesplitedge] (V2.north east) to[in=180, out=-25] (splitpoint1);
\draw[truesplitedge] (V2.south east) to[in=180, out=-25] (splitpoint3);
\draw[truesplitedge] (V3.west) to (splitpoint2);
\draw[falsesplitedge] (V2b.north east) to[in=180, out=25] (splitpoint1);
\draw[falsesplitedge] (V2b.south east) to[in=180, out=25] (splitpoint3);
\draw[truehalfedge] (splitpoint1) to (V3.north west);
\draw[falsehalfedge] (splitpoint1) to (V3.north west);
\draw[truehalfedge] (splitpoint3) to (V3.south west);
\draw[falsehalfedge] (splitpoint3) to (V3.south west);

\draw[properredge] (splitpoint2) to[out=180, in=-25] (V2.east);
\draw[properredge] (splitpoint2) to[out=180, in=25] (V2b.east);

\draw[properredge] (V3.north east) to (V4.north west);
\draw[properredge] (V3.south east) to (V4.south west);
\draw[properredge] (V5.north west) to (V4.north east);
\draw[properredge] (V4.south east) to (V5.south west);
\draw[properredge] (V5.west) to[bend right] (V3.east);
\end{tikzpicture}
\caption{The adversarial strategy used to prove Theorem~\ref{thm:signal-diameter-negative-result} on a graph with $D = 5$, during segment $3$. Regions with solid black boundaries represent sets of actual parties, with double boundaries indicating sets of parties who do not know whether $x$ or $x'$ is the true input. Regions with dashed red boundaries represent sets of imaginary parties. Solid black arrows indicate channels controlled by actual parties; dashed red arrows indicate channels controlled by imaginary parties.}
\label{fig:signal-diameter-negative-result-segment-3}
%\end{figure}

%\begin{figure}[h]
%\centering
\vspace{0.3in}

\begin{tikzpicture}
\node[truecircle] (V0) at (0, 2) {$V_0$};
\node[trueellipse] (V1) at (3, 2) {$V_1$};
\node[trueellipse] (V2) at (6, 2) {$V_2$};
\node[trueellipse] (V3) at (9, 2) {$V_3$};

\node[falsecircle] (V0b) at (0, -2) {$\bar{V}_0$};
\node[falseellipse] (V1b) at (3, -2) {$\bar{V}_1$};
\node[falseellipse] (V2b) at (6, -2) {$\bar{V}_2$};
\node[falseellipse] (V3b) at (9, -2) {$\bar{V}_3$};

\node[confusedellipse] (V4) at (12, 0) {$V_4$};
\node[confusedellipse] (V5) at (15, 0) {$V_5$};

\coordinate (splitpoint1) at ($(V4.north west) + (-0.3, 0)$);
\coordinate (splitpoint2) at ($(V5.west) + (-0.3, 0)$);
\coordinate (splitpoint3) at ($(V4.south west) + (-0.3, 0)$);

\draw[trueedge] (V0.north east) to (V1.north west);
\draw[trueedge] (V1.west) to (V0.east);
\draw[trueedge] (V0.south east) to (V1.south west);
\draw[trueedge] (V1.north east) to (V2.north west);
\draw[trueedge] (V2.south west) to (V1.south east);

\draw[falseedge] (V0b.north east) to (V1b.north west);
\draw[falseedge] (V1b.west) to (V0b.east);
\draw[falseedge] (V0b.south east) to (V1b.south west);
\draw[falseedge] (V1b.north east) to (V2b.north west);
\draw[falseedge] (V2b.south west) to (V1b.south east);

\draw[trueedge] (V2.north east) to (V3.north west);
\draw[trueedge] (V2.south east) to (V3.south west);
\draw[trueedge] (V3.west) to (V2.east);
\draw[falseedge] (V2b.north east) to (V3b.north west);
\draw[falseedge] (V2b.south east) to (V3b.south west);
\draw[falseedge] (V3b.west) to (V2b.east);

\draw[truesplitedge] (V3.north east) to[out=-25, in=180] (splitpoint1);
\draw[truesplitedge] (V3.south east) to[out=-25, in=180] (splitpoint3);
\draw[falsesplitedge] (V3b.north east) to[out=25, in=180] (splitpoint1);
\draw[falsesplitedge] (V3b.south east) to[out=25, in=180] (splitpoint3);

\draw[truehalfedge] (splitpoint1) to (V4.north west);
\draw[falsehalfedge] (splitpoint1) to (V4.north west);
\draw[truehalfedge] (splitpoint3) to (V4.south west);
\draw[falsehalfedge] (splitpoint3) to (V4.south west);

\draw[properredge] (V5.north west) to (V4.north east);
\draw[properredge] (V4.south east) to (V5.south west);
\draw[truesplitedge] (V5.west) to (splitpoint2);

\draw[truesplitedge] (splitpoint2) to[out=180, in=0] (V3.east -| V4.north);
\draw[truesplitedge] (splitpoint2) to[out=180, in=0] (V3b.east -| V4.south);
\draw[properredge] (V3b.east -| V4.south) to (V3b.east);
\draw[properredge] (V3.east -| V4.north) to (V3.east);
\end{tikzpicture}
\caption{Segment 4 of the situation depicted in Figure~\ref{fig:signal-diameter-negative-result-segment-3}.}
\label{fig:signal-diameter-negative-result-segment-4}
\end{figure}

\section{Proof of Theorem~\ref{thm:range-negative-result}} \label{apx:range-negative-result}

\subsubsection{Description of the adversary's strategy}

Fix some digraph $G$ with no isolated vertices and with chain-length $R$. We begin with the following lemma.

\begin{lem} \label{lem:exploration-number-short-walk}
There is a walk through $G$ of length $< n^2$ which visits $R$ distinct vertices, with each vertex visited at most $R$ times.
\end{lem}

\begin{proof}
Let $v_1, \dots, v_R$ be the $R$ distinct vertices visited in some chain-length walk through $G$, in the order in which they are visited. There is a path from $v_i$ to $v_{i + 1}$ of length no more than $n$ for each $1 \leq i < R$ which visits each vertex at most once. Chaining these paths together yields a walk with the desired properties.
\end{proof}

Fix some positive integer $T$. Say $\tilde{T}$ is the round complexity of the black-box simulation $\tilde{\pi}^* = C(\pi^*[G, T])$. Let $(P_{k_1}, P_{k_2}, \dots, P_{k_\ell})$ be the sequence of parties visited (with repetition) in the walk guaranteed by Lemma~\ref{lem:exploration-number-short-walk}, so that $(P_{k_i}, P_{k_{i + 1}}) \in E$ for all $i$ and $\ell \leq n^2$. Let $f_i$ denote the number of times that $P_{k_i}$ is visited in this walk, so that $1 \leq f_i \leq R$. The adversary's strategy is given by Algorithm~\ref{adversary-strategy}.

\begin{algorithm}
\caption{The strategy of the adversary $\mathcal{A}$ used to prove Theorem~\ref{thm:range-negative-result}.} \label{adversary-strategy}
\eIf{$\tilde{T} \geq R^2$}
{
  Divide the $\tilde{T}$ rounds into $\ell$ segments, with the $i$th segment containing no more than $\left\lceil \frac{\tilde{T}}{R f_i} \right\rceil$ rounds. \label{line:segmentation} \;
  \For{$i = 1$ \KwTo $\ell$}
  { \label{adversary-for-loop}
    Spend the $i$th segment zeroing out all messages going into or out of $P_{k_i}$.
  }
}
{
Do nothing.
}
\end{algorithm}
Note that the step on line~\ref{line:segmentation} is well defined, because
\begin{equation}
\sum_{i = 1}^\ell \left\lceil \frac{\tilde{T}}{R f_i} \right\rceil \geq \frac{\tilde{T}}{R} \sum_{i = 1}^\ell \frac{1}{f_i} = \tilde{T},
\end{equation}
where the last equation holds because each of $R$ distinct parties contributes a total of $1$ to the sum.

\subsubsection{Analysis of the adversary's strategy}

\begin{lem}
The adversary $\mathcal{A}$ described by Algorithm~$\ref{adversary-strategy}$ introduces a per-edge error rate of no more than $4 / R$.
\end{lem}

\begin{proof}
In the case $\tilde{T} < R^2$, the statement is trivial, so assume $\tilde{T} \geq R^2$. The number of rounds in which $\mathcal{A}$ attacks an edge of the form $(P_{k_i}, P_{k_j})$ is no more than
\begin{equation}
f_i \left\lceil \frac{\tilde{T}}{R f_i} \right\rceil + f_j \left\lceil \frac{\tilde{T}}{R f_j} \right\rceil \leq f_i + \frac{\tilde{T}}{R} + f_j + \frac{\tilde{T}}{R} \leq 2R + \frac{2\tilde{T}}{R} \leq \frac{4\tilde{T}}{R}.
\end{equation}
For an edge $e$ with an endpoint which is not a $P_{k_i}$, $\mathcal{A}$ attacks $e$ in even fewer rounds than this.
\end{proof}

Suppose $P_i$ is a party with indegree $d_i^-$ and outdegree $d_i^+$. Observe that we can identify a $T$-round transmission function $x_i$ for $P_i$ with a $(2^{d_i^-})$-ary tree of depth $T$ whose vertices are labeled with strings in $\{0, 1\}^{d_i^+}$. The edges in a path from the root to a vertex $v$ in this tree specify a sequence of bits received from each in-neighbor, and the label of $v$ specifies the bits to send in the scenario described by that path. In these terms, when a party $P_i$ in a black-box simulation $\tilde{\pi}^*$ makes a query, she effectively specifies a node in $x_i$ and asks what its label is. For an input $x = (x_1, \dots, x_n)$, we can identify in each $x_i$ the ``true path'' $t_i^x$ from the root to a leaf of $x_i$, consisting of all the edges that would be taken if the parties followed $\pi^*$ in a noiseless network. In these terms, the goal of the protocol is for each $P_i$ to learn $t_i^x$.

\begin{lem} \label{lem:quick-simulator}
Suppose $\tilde{T} < R^2$. Then if we pick an input $x$ uniformly at random, the probability that $\tilde{\pi}^*$ fails in the presence of $\mathcal{A}$ on $x$ is at least $1 - 2^{d_1^-(R^2 - T)}$, where $d_1^-$ is the indegree of $P_1$.
\end{lem}

\begin{proof}
Consider fixing an arbitrary transmission function $x_1$ and choosing the rest of $x$ uniformly at random. Then $P_1$ needs to choose between $(2^{d_1^-})^T$ possible true paths (each of which is a priori equally likely), based on $< (2^{d_1^-})^{R^2}$ bits. Thus, the probability of success is no more than $2^{d_1^-(R^2 - T)}$.
\end{proof}

\newcommand{\csq}{\left\lceil\sqrt{T}\right\rceil}
\begin{defn}
In the case that $\tilde{T} \geq R^2$, for $1 \leq i \leq \ell$, say that $E_i$ is the event that at the end of segment $(i - 1)$, the following condition holds. Let $u$ denote the node at depth $(i - 1)\csq$ along $t_{k_i}^x$. Then the subtree hanging from $u$ is completely unexplored, i.e. $P_{k_i}$ has not made any queries about the labels of any vertices in that subtree.
\end{defn}

\begin{lem} \label{lem:no-exploration}
Suppose $\tilde{T} \geq R^2$ and $C(\pi^*)$ and has query complexity no more than $2^{\sqrt{T}}(1 - \delta)$, where $0 \leq \delta \leq 1$. Suppose we pick an input $x$ uniformly at random, and execute $\tilde{\pi}^*$ on $x$ in the presence of $\mathcal{A}$. Then for each $1 \leq i \leq \ell$, conditioned on $E_1, \dots, E_{i - 1}$, the probability that $E_i$ occurs is at least $\delta$.
\end{lem}

\begin{proof}
For the analysis, it suffices to fix arbitrary values for any random bits that $\tilde{\pi}^*$ uses (still choosing $x$ randomly.) Vacuously, $E_1$ occurs with probability 1, so assume $i > 1$.  Consider an arbitrary input $x$ such that $E_1, \dots, E_{i - 1}$ occur. Let $w$ be the node at depth $(i - 2)\csq$ in $t_{k_{i - 1}}^x$, and let $\tau$ be the subtree hanging from $w$. Since $E_{i - 1}$ occurred, at the beginning of segment $(i - 1)$, $P_{k_{i - 1}}$ had not explored any of $\tau$. Therefore, if $x'$ is the same as $x$ except for the labels of the nodes in $\tau$, then the execution of $\tilde{\pi}^*(x')$ prior to segment $(i - 1)$ is the same as that of $\tilde{\pi}^*(x)$. Let $U$ denote the set of locations of labels of nodes in $\tau$ which correspond to bits transmitted from $P_{k_{i - 1}}$ to $P_{k_i}$. Consider altering $x$ by assigning values to the bits in $U$ uniformly at random.

Starting at depth $(i - 2)\csq$ in $t_{k_i}^x$, at each level, the true path could go one of two ways, depending on the value of a bit in $U$ which is at the same level in $x_{k_i}$. Thus, the node at depth $i \csq$ in $t_{k_i}^x$ could be any of $2^{\csq}$ different nodes, each with equal probability. Say this set of $2^{\csq}$ nodes is $S$. At the end of segment $(i - 1)$, $P_{k_i}$ has made fewer than $2^{\sqrt{T}}(1 - \delta)$ queries total, and thus the fraction of nodes in $S$ whose subtrees she has not explored at all is at least $\delta$. The queries she chooses to make during segment $(i - 1)$ cannot depend on the labels assigned to nodes in $U$, because $P_{k_{i - 1}}$ is attacked by $\mathcal{A}$ during segment $(i - 1)$. Therefore, when we assign values to $U$ uniformly at random, the probability that $E_i$ occurs is at least $\delta$. Therefore, if we choose $x$ uniformly at random, then conditioned on $E_1, \dots, E_{i - 1}$, the probability that $E_i$ occurs is at least $\delta$.
\end{proof}

\begin{lem} \label{polynomial-lemma}
Suppose $\mathcal{Q}(T)$ is a polynomial. Then for all sufficiently large $T$,
\begin{equation}
\mathcal{Q}(T) < 2^{\sqrt{T}}\left(1 - 2^{-1/T}\right).
\end{equation}
\end{lem}

\begin{proof}
Note that $\frac{1}{2}e^{1/2} < 1$. Therefore, from the limit definition of the exponential function, we see that for sufficiently large $T$,
\begin{equation}
\left(1 - \frac{1/2}{T}\right)^{T} > \frac{1}{2} e^{1/2} \cdot e^{-1/2} = \frac{1}{2}.
\end{equation}
Taking a $T$th root of both sides gives $1 - \frac{1}{T} > 2^{-1/T}$, and therefore
\begin{equation}
2^{\sqrt{T}}\left(1 - 2^{-1/T}\right) > \frac{2^{\sqrt{T}}}{T}.
\end{equation}
Obviously, for sufficiently large $T$, the right-hand side is larger than $\mathcal{Q}(T)$.
\end{proof}

\begin{proof}
[of Theorem~$\ref{thm:range-negative-result}$] Say that the simulations output by $C$ make a number of queries which is bounded by the polynomial $\mathcal{Q}(T)$. Since we only care about the limit as $T \to \infty$, we may assume that $T > 16n^4$, and by Lemma~\ref{polynomial-lemma}, we may also assume that
\begin{equation}
\mathcal{Q}(T) < 2^{\sqrt{T}} \left(1 - 2^{-1/T}\right).
\end{equation}
We will show that if we pick an input $x$ for $\pi^*[G, T]$ uniformly at random, the failure probability $\delta$ of $\tilde{\pi}^*$ in the presence of $\mathcal{A}$ satisfies
\begin{equation}
\delta \geq \min\{2^{-(n^2/T)} \cdot (1 - 2^{-\frac{1}{2}T}), (1 - 2^{d_1^-(R^2 - T)})\} \label{goal-equation}
\end{equation}
which in particular means that $\delta \to 1$ as $T \to \infty$. (Recall that $G$ is fixed.) If $\tilde{T} < R^2$, then we are done by Lemma~\ref{lem:quick-simulator}. Assume, therefore, that $\tilde{T} \geq R^2$.

The probability that $E_i$ happens for all $1 \leq i \leq \ell$ is $\prod_{j = 1}^{\ell} \Pr(E_i | E_1, \dots, E_{i - 1})$, which by Lemma~\ref{lem:no-exploration} is at least $2^{-\ell/T}$. Since $\ell \leq n^2$, the probability that $E_j$ happens for all $1 \leq i \leq \ell$ is at least $2^{-n^2/T}$. Suppose $E_{\ell - 1}$ happens. Then at the end of segment $(\ell - 2)$, $P_{k_{\ell - 1}}$ has not made any queries about the labels of any of the vertices in the subtree $\tau$ hanging from the node at depth $(\ell - 2)\csq$ along $t_{k_{\ell - 1}}^x$. The height $h$ of $\tau$ satisfies
\begin{equation}
h = T - (\ell - 2)\csq \geq T - n^2\csq \geq T - 2n^2\sqrt{T}.
\end{equation}
Since $T > 16n^4$, we have $\frac{1}{2}\sqrt{T} > 2n^2$, and hence $h \geq \frac{1}{2}T$.

During segment $\ell$, all messages going into $P_{k_{\ell}}$ are zeroed out, so the output of $P_{k_\ell}$ at the end of the protocol does not depend on queries that $P_{k_{\ell - 1}}$ makes during segment $\ell$. After fixing the labels of all nodes other than those in $\tau$, each step in $t_{k_\ell}^x$ at the level of $\tau$ could go one of two ways, based on a label of a node in $\tau$. Therefore, conditioned in $E_{\ell - 1}$, the probability that $P_{k_{\ell}}$ correctly guesses $t_{k_{\ell}}^x$ is no more than $2^{-\frac{1}{2}T}$, since $\tau$ has height at least $\frac{1}{2}T$. Thus, conditioned on $E_{\ell - 1}$, the probability that $P_{k_{\ell' + 1}}$ fails to guess her transcript is at least $(1 - 2^{-\frac{1}{2}T})$. Multiplying, we see that (unconditionally) the probability that $P_j$ fails to guess her transcript is at least $2^{-n^2/T}(1 - 2^{-\frac{1}{2}T})$, and thus Equation~\ref{goal-equation} is satisfied.
\end{proof}

\end{document}